\documentclass[11pt]{article}

\usepackage[margin=1in]{geometry}
\usepackage{amsmath,amssymb,amsfonts,amsthm,mathtools}
\usepackage{graphicx}
\usepackage{booktabs}
\usepackage[caption=false]{subfig}
\usepackage[numbers,sort&compress]{natbib}
\bibpunct[, ]{[}{]}{,}{n}{,}{,}

\usepackage[hidelinks]{hyperref}
\usepackage[nameinlink,capitalize]{cleveref}

\theoremstyle{plain}
\newtheorem{theorem}{Theorem}[section]
\newtheorem{lemma}[theorem]{Lemma}
\newtheorem{corollary}[theorem]{Corollary}
\newtheorem{proposition}[theorem]{Proposition}

\theoremstyle{definition}
\newtheorem{assumption}[theorem]{Assumption}
\newtheorem{definition}[theorem]{Definition}

\theoremstyle{remark}
\newtheorem{remark}[theorem]{Remark}

\crefname{assumption}{Assumption}{Assumptions}
\Crefname{assumption}{Assumption}{Assumptions}
\crefname{definition}{Definition}{Definitions}
\Crefname{definition}{Definition}{Definitions}

\DeclareMathOperator*{\argmax}{arg\,max}
\DeclareMathOperator{\KL}{KL}
\DeclareMathOperator{\dist}{dist}
\newcommand{\R}{\mathbb{R}}
\newcommand{\E}{\mathbb{E}}
\newcommand{\Prob}{\mathbb{P}}
\newcommand{\bbF}{\mathbb{F}}
\newcommand{\cF}{\mathcal{F}}
\newcommand{\cQ}{\mathcal{Q}}
\newcommand{\cP}{\mathcal{P}}
\newcommand{\cA}{\mathcal{A}}
\newcommand{\one}{\mathbf{1}}
\newcommand{\norm}[1]{\left\lVert #1\right\rVert}
\newcommand{\abs}[1]{\left\lvert #1\right\rvert}

\title{Entropy-Regularized Certainty-Equivalent Bellman Policies for Risk-Sensitive Market Making}
\author{Tenghan Zhong\\
Department of Mathematics, University of Southern California, Los Angeles, CA, USA\\
\texttt{tenghanz@usc.edu}}
\date{}

\begin{document}
\maketitle\begin{abstract}
We study a finite-inventory risk-sensitive market making problem in which a
dealer controls bid and ask quotes, faces Brownian midprice risk, and receives
liquidity-taking orders through point processes with quote-dependent
intensities. The objective is the certainty equivalent induced by exponential
utility with terminal and running inventory penalties. We introduce an exact
discrete entropy-regularized Bellman operator that applies log-sum-exp
regularization to deterministic-action certainty-equivalent scores, rather than
to a risk-neutral one-step reward. This distinction is essential because the
exponential certainty equivalent does not commute with quote randomization.

For time step \(h\) and entropy parameter \(\lambda\), we prove uniform
convergence to the unregularized continuous-time risk-sensitive value at rate
\[
    O\bigl(h+\lambda(1+|\log\lambda|)\bigr).
\]
We also prove certainty-equivalent performance bounds for the induced Gibbs
policies under a fresh-sampling relaxed implementation, in which quote marks are
sampled at potential fill events rather than frozen over a time step. Under a
quadratic growth condition on the Hamiltonian in the relevant quote coordinates,
these policies concentrate around the unregularized optimal quote set. Finally,
we show that a lower-cost Hamiltonian-Gibbs proxy satisfies a
certainty-equivalent performance bound of the same order as the exact Bellman
Gibbs policy. Numerical experiments in an Avellaneda--Stoikov specification
support the predicted scaling for discretization error, entropy bias, policy
gap, quote concentration, and exact-versus-proxy consistency.
\end{abstract}



\section{Introduction}
\label{sec:introduction}

Market makers provide liquidity by continuously posting bid and ask quotes. Their
profit comes from spread capture, but this profit is exposed to execution
uncertainty, adverse inventory accumulation, and midprice risk. A quote posted
too close to the midprice increases execution frequency but also increases the
risk of accumulating an undesirable position. A quote posted too far away reduces
inventory pressure but also sacrifices fills. This tradeoff is the central
mechanism in stochastic-control models of market making
\cite{HoStoll1981,AvellanedaStoikov2008,GueantLehalleFernandezTapia2013,CarteaJaimungalPenalva2015}.

The canonical continuous-time formulation is the Avellaneda--Stoikov model
\cite{AvellanedaStoikov2008}. In that model, the midprice is driven by a
Brownian motion, buy and sell market orders arrive according to quote-dependent
Poisson intensities, and the dealer maximizes an exponential-utility criterion.
Subsequent work, especially Gu\'eant, Lehalle, and Fernandez-Tapia
\cite{GueantLehalleFernandezTapia2013}, showed how finite-inventory versions of
this problem can be reduced to systems of ordinary differential equations and
how the resulting optimal quotes can be characterized under inventory
constraints. These models are analytically tractable and financially
interpretable, which is why they remain a useful benchmark for both
stochastic-control and reinforcement-learning approaches to market making.

This paper revisits this finite-inventory risk-sensitive market-making problem
from the perspective of entropy-regularized Bellman policies. Entropy
regularization has become a standard device in reinforcement learning
\cite{Ziebart2010,GeistScherrerPietquin2019,HaarnojaEtAl2018} because it
replaces hard maximization by a smooth log-sum-exp operator and produces Gibbs
policies rather than deterministic selectors. In a risk-neutral Markov decision
problem, this operation is relatively transparent: one adds an entropy bonus to
the reward or, equivalently, penalizes relative entropy with respect to a
reference action law. In a risk-sensitive problem, however, randomization is more
delicate. The objective is not a linear expectation of rewards, but an outer
exponential certainty equivalent. Therefore, one must distinguish between
regularizing deterministic-action certainty-equivalent scores and first
randomizing a quote before applying the outer certainty equivalent. These two
operations do not commute.

This paper develops a model-based CE-score Bellman approximation theory for
randomized quoting policies in a finite-inventory point-process market-making
model. Its object is the exact one-step certainty-equivalent Bellman operator
and the induced Gibbs quoting policy, rather than a model-free \(q\)-learning
algorithm for exploratory diffusions. The main contribution of the paper is an
exact discrete entropy-regularized Bellman construction that respects this
distinction. For a fixed deterministic quote \(\delta\) over one
time step, we define the one-step certainty-equivalent score
\(C_h^\delta\varphi(q)\). We then apply log-sum-exp regularization to these
deterministic-action scores. In the following display, \(h>0\) is the time step,
\(\lambda>0\) is the entropy parameter, \(A\) is the quote set, \(\nu\) is the
reference quote law, \(\varphi\) is a continuation value, and \(q\) is the
current inventory state:
\[
    (T_h^\lambda\varphi)_q
    =
    h\lambda
    \log\int_A
    \exp\left(\frac{C_h^\delta\varphi(q)}{h\lambda}\right)\nu(d\delta).
\]
The normalization by \(h\lambda\) is essential. It makes the discrete operator
consistent with the entropy-regularized Hamiltonian for the
deterministic-action Hamiltonian \(H_q(y,\delta)\) and a value vector \(y\):
\[
    H_q^\lambda(y)
    =
    \lambda\log\int_A
    \exp\left(\frac{H_q(y,\delta)}{\lambda}\right)\nu(d\delta)
\]
in the continuous-time limit. The operator \(T_h^\lambda\) is therefore not an
ad hoc softmax approximation; it is an exact one-step certainty-equivalent
Bellman operator followed by entropy regularization at the score level.

Our first result is a quantitative hard-limit theorem. Let \(v^{h,\lambda}\) be
the discrete value generated by \(T_h^\lambda\), and let \(v^0\) be the hard
continuous-time risk-sensitive value. With grid times \(t_n=nh\), terminal index
\(N\), and finite inventory set \(\cQ\), we prove in \Cref{thm:main} that
\[
    \max_{0\le n\le N}\max_{q\in\cQ}
    |v_{n,q}^{h,\lambda}-v_q^0(t_n)|
    \le
    C\bigl[h+\lambda(1+|\log\lambda|)\bigr].
\]
The \(O(h)\) term is the time-discretization error. The
\(\lambda(1+|\log\lambda|)\) term is the entropy bias. The logarithmic factor is
not a numerical artifact; it comes from lower bounding the mass of a small
neighborhood of an optimizer in the two-dimensional quote rectangle. The bound
is uniform over time and inventory states.

The second contribution is a policy-performance theorem. The Gibbs density
associated with \(T_h^\lambda\) is a randomized quote policy. Because the
risk-sensitive objective is nonlinear, its continuous-time interpretation must
be specified carefully. We use a fresh-sampling relaxed-control implementation:
at each potential ask or bid fill, the quote mark is sampled from the current
state-dependent policy, and the marked point-process compensator is averaged
over this policy. Writing \(V^*\) for the hard value, \(J(\cdot;\pi)\) for the
certainty equivalent under a policy \(\pi\), and \(\pi^{\rm ex}\) for the exact
Bellman Gibbs policy, we prove the performance gap
\[
    0\le V^*(t,x,s,q)-J(t,x,s,q;\pi^{\rm ex})
    \le
    C\bigl[h+\lambda(1+|\log\lambda|)\bigr].
\]
This result is not merely a value-convergence statement. It shows that the
randomized policy induced by the discrete entropy-regularized Bellman operator
attains an \(O(h+\lambda(1+|\log\lambda|))\) certainty-equivalent performance
gap under the corresponding continuous-time relaxed market-making
implementation.

The third contribution connects the exact Bellman policy to a lower-cost
Hamiltonian-Gibbs proxy. Computing \(C_h^\delta\varphi(q)\) exactly for many
quotes is more expensive than evaluating the Hamiltonian \(H_q(\varphi,\delta)\).
For this reason, with \(\widehat v^{h,\lambda}\) denoting the Euler
discretization of the soft Hamiltonian ODE, the computationally practical policy
used in our large-scale experiments is the proxy
\[
    \pi_{n,q}^{\rm Ham}(d\delta)
    \propto
    \exp\left(
        \frac{H_q(\widehat v_{n+1}^{h,\lambda},\delta)}{\lambda}
    \right)\nu(d\delta),
\]
We prove that this proxy satisfies the same fresh-sampling performance bound as
the exact Bellman Gibbs policy. Thus the computationally practical policy has
its own certainty-equivalent performance guarantee, rather than only an ex post
comparison with the exact Bellman Gibbs policy.

The fourth contribution is a quote-concentration result. Under an
active-coordinate quadratic growth condition, the Gibbs policies concentrate
around the hard optimal quote set. More precisely, the expected squared
active-coordinate distance between a randomized quote and the hard optimizer is
bounded by the same scale
\[
    C\bigl[h+\lambda(1+|\log\lambda|)\bigr].
\]
For exponential arrival intensities, we give a verifiable curvature certificate
that implies this quadratic growth condition. This result translates value and
policy convergence into a direct statement about quoted bid and ask spreads.

Finally, the numerical experiments are designed to test the theory rather than
only to display profitability. We separately test the \(O(h)\) discretization
component, the entropy-bias scale, the fresh-sampling policy-performance gap,
the active-coordinate quote-concentration bound, and the consistency between the
exact Bellman Gibbs policy and the Hamiltonian-Gibbs proxy. The observed rates
match the predicted scales and support the use of the Hamiltonian-Gibbs policy
as a scalable implementation of the exact Bellman construction.

The paper is organized as follows. \Cref{sec:related-work} reviews related work
on stochastic-control market making, risk-sensitive control,
entropy-regularized reinforcement learning, relaxed-control limits, and
reinforcement learning for market making. \Cref{sec:model-value} formulates the
finite-inventory risk-sensitive market-making model, defines strict admissible
systems, and introduces the hard and soft Hamiltonian systems.
\Cref{sec:ce-score-scheme} introduces the exact discrete CE-score entropy
Bellman operator and its variational representation. \Cref{sec:main-convergence}
states the main convergence theorem and gives the proof strategy; the detailed
auxiliary estimates, verification, entropy-bias argument, finite-state
Feynman--Kac formula, and one-step consistency proof are collected in
\Cref{app:auxiliary,app:soft-hard,app:fk-consistency}. \Cref{sec:fresh-sampling}
states the fresh-sampling performance bounds for the exact Bellman Gibbs policy
and the Hamiltonian-Gibbs proxy, with proofs in \Cref{app:policy-proofs}.
\Cref{sec:quote-concentration} establishes quote concentration and gives a
curvature certificate for exponential arrival intensities, with proof details in
\Cref{app:quote-proofs}. \Cref{sec:numerics} presents numerical experiments that
validate the value, policy, quote-concentration, and exact-vs-proxy consistency
bounds. \Cref{sec:conclusion} concludes and discusses extensions.

\section{Related Work}
\label{sec:related-work}

\paragraph{Stochastic-control models of market making.}
The mathematical analysis of dealer pricing and inventory control begins with
Ho and Stoll \cite{HoStoll1981}, who studied the dealer's optimal bid and ask
prices under inventory and return uncertainty. Avellaneda and Stoikov
\cite{AvellanedaStoikov2008} introduced a tractable high-frequency
limit-order-book model in which the dealer controls bid and ask spreads, the
midprice follows a Brownian motion, and market order arrivals are modeled by
Poisson processes whose intensities depend on the quoted distances from the
midprice. Gu\'eant, Lehalle, and Fernandez-Tapia
\cite{GueantLehalleFernandezTapia2013} developed a finite-inventory version of
this framework and showed that the HJB system can be transformed into a system
of linear ODEs under exponential utility. Cartea, Jaimungal, and Penalva
\cite{CarteaJaimungalPenalva2015} provide a broader stochastic-control
treatment of algorithmic and high-frequency trading, including market making,
optimal execution, and inventory-aware quoting.

Several extensions of the classical model incorporate richer state dynamics,
signals, market impact, and latency. Fodra and Labadie
\cite{FodraLabadie2012} consider market making with inventory constraints and
directional bets under more general midprice dynamics. Cartea, Jaimungal, and
Ricci \cite{CarteaJaimungalRicci2014} introduce mutually exciting order-flow
features into high-frequency trading models. Cartea and Wang
\cite{CarteaWang2020} study market making with alpha signals, showing how
short-term information can be used to manage adverse selection and inventory
risk. Cartea and S\'anchez-Betancourt \cite{CarteaSanchezBetancourt2021}
analyze latency and the shadow price of improving fill ratios. More recent
work, such as Ch\'avez-Casillas, Figueroa-L\'opez, Yu, and Zhang
\cite{ChavezCasillasEtAl2024}, studies adaptive optimal market-making
strategies with inventory liquidation costs and data-driven demand functions.

The present paper remains close to the analytically transparent
finite-inventory Avellaneda--Stoikov setting, but it adds a different layer: a
risk-sensitive entropy-regularized Bellman policy with quantitative
soft-to-hard, discrete-to-continuous, and policy-performance guarantees. The
main novelty is therefore not a new microstructure state variable, but the
construction and analysis of an entropy-regularized Bellman policy that is
compatible with the nonlinear certainty-equivalent objective.

\paragraph{Risk-sensitive stochastic control.}
The exponential utility criterion used in market-making models is a standard
risk-sensitive control objective. Classical risk-sensitive control studies show
that exponential criteria lead to nonlinear dynamic programming equations and
are closely related to robustness under model uncertainty
\cite{Whittle1990,FlemingMcEneaney1995,BieleckiPliska1999}. In the present
setting, this nonlinearity is responsible for a key technical issue:
randomization before the outer exponential certainty equivalent is not the same
as entropy regularization of deterministic-action certainty-equivalent scores.
Our discrete operator is designed around this distinction. The proof therefore
uses risk-sensitive one-step certainty equivalents and exponential
semimartingale arguments, rather than only risk-neutral reward averaging.

\paragraph{Entropy-regularized reinforcement learning and soft Bellman
operators.}
Entropy regularization has a long history in maximum entropy and causal entropy
formulations of decision making \cite{Ziebart2010}, in soft and regularized
Bellman operators \cite{FoxPakmanTishby2016,HaarnojaEtAl2017,HaarnojaEtAl2018},
and in the general theory of regularized Markov decision processes
\cite{GeistScherrerPietquin2019}. These works show that entropy regularization
smooths the hard maximization step and produces stochastic policies with
Boltzmann or Gibbs form. In discrete-time risk-neutral settings, this leads to a
clean equivalence between entropy-augmented rewards, convex duality, and soft
Bellman equations.

Continuous-time entropy-regularized reinforcement learning requires additional
care because the time step tends to zero while the exploration distribution
remains nondegenerate. Wang, Zariphopoulou, and Zhou
\cite{WangZariphopoulouZhou2020} formulate reinforcement learning in continuous
time and space through exploratory stochastic control and entropy-regularized
relaxed controls. Tang, Zhang, and Zhou \cite{TangZhangZhou2022} study
exploratory HJB equations and their convergence. Jia and Zhou
\cite{JiaZhou2023} develop continuous-time \(q\)-learning under the
entropy-regularized exploratory diffusion formulation. Jia
\cite{Jia2026RiskSensitive} extends this direction to continuous-time
risk-sensitive reinforcement learning with an exponential-form objective. That
work develops a \(q\)-learning and martingale characterization for
entropy-regularized exploratory diffusions, where risk sensitivity introduces a
quadratic-variation correction to the value process. Xie
\cite{Xie2025RiskSensitiveQLearning} is a close conceptual neighbor: it studies
continuous-time risk-sensitive \(q\)-learning for controlled diffusions with
nonlinear cumulative-reward functionals and uses optimized certainty
equivalents to recover Markovian structure in an augmented environment.

Our paper is related to this continuous-time entropy-regularized control
literature, but the model, operator, and theorem are different. We do not
develop a model-free \(q\)-learning algorithm, an OCE-augmented diffusion
formulation, or a martingale characterization for exploratory diffusions.
Instead, we study a model-based finite-inventory market-making problem in which
the inventory is driven by controlled point-process fills. Moreover, the entropy
regularization is applied to one-step certainty-equivalent scores generated by
deterministic frozen quotes. This is why the exact discrete Bellman operator is
\[
    h\lambda\log\int_A
    \exp\{C_h^\delta\varphi(q)/(h\lambda)\}\nu(d\delta),
\]
rather than a direct discretization of an exploratory diffusion reward. The
proof also needs a finite-state Feynman--Kac representation for the
risk-sensitive jump-inventory dynamics.

\paragraph{Mixed policies, relaxed controls, and discretization limits.}
The connection between randomized discrete-time policies and continuous-time
relaxed controls has recently received renewed attention. Carmona and
Lauri\`ere \cite{CarmonaLauriere2025} rigorously connect discrete-time mixed
policies with continuous-time relaxed controls as the time mesh tends to zero.
Jia, Ouyang, and Zhang \cite{JiaOuyangZhang2025} analyze discretely sampled
stochastic policies in continuous-time reinforcement learning and quantify
convergence to dynamics with coefficients aggregated under the stochastic
policy. Pham, Zhang, and Zhu \cite{PhamZhangZhu2026} study discretization errors
from regularized discrete-time reinforcement learning to continuous-time
stochastic control.

The fresh-sampling policy evaluation in this paper is aligned with this
literature but specialized to market making with marked point-process fills.
The compensator of the marked ask and bid fill processes is averaged over the
current quote distribution. This specification is essential because it matches
the Hamiltonian average appearing in the relaxed policy-evaluation ODE. We
therefore do not claim that a policy obtained by sampling one quote at the
beginning of a time interval and freezing it under an outer risk-sensitive
certainty equivalent has the same performance guarantee. The result is stated
for the fresh-sampling relaxed implementation precisely to avoid this
identification error.

\paragraph{Reinforcement learning for market making.}
A growing computational literature applies reinforcement learning to market
making. Spooner, Fearnley, Savani, and Koukorinis
\cite{SpoonerEtAl2018} use temporal-difference learning in a high-fidelity
limit-order-book simulator. Lim and Gorse \cite{LimGorse2018} study
reinforcement learning for high-frequency market making with state and reward
features tailored to inventory control. Spooner and Savani
\cite{SpoonerSavani2020} use adversarial reinforcement learning to improve
robustness to model uncertainty. Closer to the continuous-time theory side,
Cao, \v{S}i\v{s}ka, Szpruch, and Treetanthiploet
\cite{CaoSiskaSzpruchTreetanthiploet2024} prove logarithmic regret bounds for
learning the liquidity sensitivity parameter in an ergodic
Avellaneda--Stoikov model, while Zheng and Ding \cite{ZhengDing2024} analyze
sampling-frequency error and complexity tradeoffs for reinforcement learning in
high-frequency market making. Recent studies consider more realistic latency,
nonstationarity, Hawkes or semi-Markov dynamics, and deep reinforcement-learning
architectures \cite{JiangEtAl2025,LalorSwishchuk2025,ZimmerCosta2025}. Fang and
Israel \cite{FangIsrael2025} propose a Wasserstein robust market-making
framework that uses stochastic policies with entropy regularization. Their
entropy term is used in a distributionally robust market-making formulation. By
contrast, our entropy term regularizes one-step risk-sensitive
certainty-equivalent scores and leads to soft-to-hard Bellman approximation,
fresh-sampling performance bounds, and quote-concentration estimates.

These works demonstrate the practical appeal of randomized and learning-based
market-making policies. However, many of them are primarily simulator-driven,
algorithmic, or robustness-oriented and do not provide a finite-inventory
risk-sensitive CE-score Bellman operator with a continuous-time convergence and
performance theorem. The present paper fills that gap in a controlled
Avellaneda--Stoikov benchmark where sharp value, policy, and quote-level
estimates can be proved. It gives explicit rates for value convergence, policy
performance, and quote concentration, and it proves that the Hamiltonian-Gibbs
policy used in large-scale computation satisfies the same
\(O(h+\lambda(1+|\log\lambda|))\) certainty-equivalent performance bound as the
exact risk-sensitive CE-score Bellman policy.

\section{Model and Risk-Sensitive Value}
\label{sec:model-value}

We work in a finite-inventory Avellaneda--Stoikov framework, in the spirit of
\cite{AvellanedaStoikov2008,GueantLehalleFernandezTapia2013}.

Fix
\[
    T<\infty,\qquad Q\in\mathbb N,\qquad
    \cQ:=\{-Q,-Q+1,\ldots,Q\}.
\]
The midprice is driven by a Brownian motion \(W\) and satisfies
\[
    dS_u=\sigma\,dW_u,\qquad \sigma>0,
\]
and the quote set is the rectangle
\[
    A=[\underline\delta,\overline\delta]^2\subset\R^2,
    \qquad \underline\delta<\overline\delta.
\]

\begin{assumption}[Standing assumptions]\label{ass:model}
The intensity functions
\[
    \Lambda^a,\Lambda^b:[\underline\delta,\overline\delta]\to[0,\infty)
\]
belong to \(C^1([\underline\delta,\overline\delta])\) and satisfy
\[
    0\le \Lambda^a,\Lambda^b\le \overline\Lambda<\infty.
\]
The parameters satisfy
\[
    \gamma>0,\qquad \Phi\ge0,\qquad \eta\ge0.
\]
The reference probability measure \(\nu\) on \(A\) has a Borel density \(m\)
with respect to Lebesgue measure such that
\[
    0<m_-\le m(\delta)\le m_+<\infty
\]
for Lebesgue-a.e. \(\delta\in A\). We fix such a Borel version of \(m\).
\end{assumption}

Throughout, \(\one_E\) denotes the indicator of an event \(E\).

\begin{definition}[Strict admissible systems]\label{def:strict-system}
Given an initial state \((t,x,s,q)\in[0,T]\times\R\times\R\times\cQ\), a
strict admissible system on \([t,T]\) is a tuple
\[
    \mathfrak S=(\Omega,\cF,\bbF,\Prob,W,N^a,N^b,\delta,X,S,q_\cdot)
\]
such that \(\bbF\) satisfies the usual conditions, \(W\) is a Brownian motion,
\(N^a,N^b\) are simple counting processes with no common jumps, and
\(\delta_u=(\delta^a_u,\delta^b_u)\) is an \(A\)-valued predictable process. Moreover,
\[
    S_u=s+\sigma(W_u-W_t),
\]
\[
    q_u=q+\int_t^u dN_r^b-\int_t^u dN_r^a,
\]
\[
    X_u=x+\int_t^u(S_r+\delta_r^a)dN_r^a
          -\int_t^u(S_r-\delta_r^b)dN_r^b.
\]
The \(\bbF\)-predictable compensators of \(N^a\) and \(N^b\) are
\[
    \int_t^u \Lambda^a(\delta_r^a)\one_{\{q_{r-}>-Q\}}\,dr,
    \qquad
    \int_t^u \Lambda^b(\delta_r^b)\one_{\{q_{r-}<Q\}}\,dr.
\]
The set of all such strict admissible systems is denoted by \(\mathfrak A_t\). When no
confusion is possible, we write \(\delta\in\mathfrak A_t\), but the admissible
object is the whole system. This avoids defining the filtration through a
counting process that is itself constructed from the control.
\end{definition}

For every strict system, integration by parts gives the self-financing identity
\begin{equation}\label{eq:self-financing}
    X_u+q_uS_u
    =x+qs+
    \int_t^u q_{r-}\,dS_r
    +\int_t^u\delta_r^a\,dN_r^a
    +\int_t^u\delta_r^b\,dN_r^b .
\end{equation}
The boundary indicators in the compensators imply \(q_u\in\cQ\) for all
\(u\in[t,T]\).

For \(\mathfrak S\in\mathfrak A_t\), let
\(\E^{\mathfrak S}_{t,x,s,q}\) denote expectation under the strict admissible
system \(\mathfrak S\) started from \((t,x,s,q)\), and define
\[
    J(t,x,s,q;\mathfrak S)
    :=
    -\frac1\gamma\log
    \E^{\mathfrak S}_{t,x,s,q}
    \left[
        \exp\left(
        -\gamma\left[
            X_T+q_TS_T-\Phi q_T^2
            -\int_t^T\eta q_u^2\,du
        \right]
        \right)
    \right].
\]

The hard value is
\begin{equation}\label{eq:hard-value}
    V^*(t,x,s,q):=\sup_{\mathfrak S\in\mathfrak A_t}J(t,x,s,q;\mathfrak S).
\end{equation}

For \(y=(y_q)_{q\in\cQ}\in\R^{2Q+1}\), define
\[
    \Delta_q^a(y,\delta)
    :=
    \begin{cases}
        \delta^a+y_{q-1}-y_q, & q>-Q,\\
        0, & q=-Q,
    \end{cases}
\]
\[
    \Delta_q^b(y,\delta)
    :=
    \begin{cases}
        \delta^b+y_{q+1}-y_q, & q<Q,\\
        0, & q=Q,
    \end{cases}
\]
and
\[
    \Lambda_q^a(\delta^a):=\Lambda^a(\delta^a)\one_{\{q>-Q\}},
    \qquad
    \Lambda_q^b(\delta^b):=\Lambda^b(\delta^b)\one_{\{q<Q\}}.
\]
The deterministic-action Hamiltonian is
\begin{equation}\label{eq:H}
\begin{aligned}
    H_q(y,\delta)
    :={}&
    -\eta q^2-\frac{\gamma\sigma^2}{2}q^2 \\
    &+
    \frac{\Lambda_q^a(\delta^a)}{\gamma}
    \left(1-e^{-\gamma\Delta_q^a(y,\delta)}\right)
    +
    \frac{\Lambda_q^b(\delta^b)}{\gamma}
    \left(1-e^{-\gamma\Delta_q^b(y,\delta)}\right).
\end{aligned}
\end{equation}
For \(\lambda>0\), set
\begin{equation}\label{eq:H0-Hlambda}
    H_q^0(y):=\sup_{\delta\in A}H_q(y,\delta),
    \qquad
    H_q^\lambda(y):=
    \lambda\log\int_A\exp\left(\frac{H_q(y,\delta)}{\lambda}\right)\nu(d\delta).
\end{equation}
The hard and soft reduced values solve
\begin{equation}\label{eq:hard-ode}
    -\dot v_q^0(t)=H_q^0(v^0(t)),
    \qquad v_q^0(T)=-\Phi q^2,
\end{equation}
\begin{equation}\label{eq:soft-ode}
    -\dot v_q^\lambda(t)=H_q^\lambda(v^\lambda(t)),
    \qquad v_q^\lambda(T)=-\Phi q^2.
\end{equation}
Since \(A\) is compact and \(H_q\) is locally Lipschitz in \(y\) uniformly over
\(\delta\in A\), both ODEs have unique local solutions. The a priori bounds in
\Cref{lem:bounds} extend them to \([0,T]\).

\section{Entropy-Regularized CE-Score Bellman Scheme}
\label{sec:ce-score-scheme}

For \(\varphi=(\varphi_q)_{q\in\cQ}\), set
\[
    F_\varphi(x,s,q):=x+qs+\varphi_q.
\]
For \(h>0\) and a deterministic quote \(\delta\in A\) frozen over one step
\([0,h]\), let \(\E_{0,0,q}^{\delta}\) denote expectation under the canonical
fixed-quote system started from \(X_0=S_0=0\) and inventory \(q\). In that
system, the Brownian motion is independent of the inventory chain, and the
chain has state-dependent transition rates
\[
    r\to r-1 \text{ at rate } \Lambda_r^a(\delta^a),
    \qquad
    r\to r+1 \text{ at rate } \Lambda_r^b(\delta^b),
    \qquad r\in\cQ .
\]
Define
\begin{equation}\label{eq:Ch}
    C_h^\delta\varphi(q)
    :=
    -\frac1\gamma
    \log
    \E_{0,0,q}^\delta
    \left[
        \exp\left(
        -\gamma\left[
            F_\varphi(X_h,S_h,q_h)
            -\int_0^h\eta q_u^2\,du
        \right]
        \right)
    \right].
\end{equation}
For \(h>0\) and \(\lambda>0\), the discrete entropy-regularized Bellman operator
is
\begin{equation}\label{eq:Th}
    (T_h^\lambda\varphi)_q
    :=
    h\lambda
    \log\int_A
    \exp\left(\frac{C_h^\delta\varphi(q)}{h\lambda}\right)\nu(d\delta).
\end{equation}
Equivalently, with \(\cP(A)\) denoting the set of probability measures on \(A\)
and \(\KL(\pi\|\nu)\) denoting relative entropy, we use the convention
\[
    \KL(\pi\|\nu)
    =
    \begin{cases}
    \displaystyle
    \int_A \frac{d\pi}{d\nu}
    \log\left(\frac{d\pi}{d\nu}\right)d\nu,
    & \pi\ll\nu,\\[0.8em]
    +\infty, & \text{otherwise}.
    \end{cases}
\]
\begin{equation}\label{eq:Th-var}
    (T_h^\lambda\varphi)_q
    =
    \sup_{\pi\in\cP(A)}
    \left\{
        \int_A C_h^\delta\varphi(q)\,\pi(d\delta)
        -h\lambda\KL(\pi\|\nu)
    \right\}.
\end{equation}
For \(t_n=nh\) and \(Nh=T\), the exact discrete CE-score value is
\begin{equation}\label{eq:discrete-value}
    v_{N,q}^{h,\lambda}=-\Phi q^2,
    \qquad
    v_n^{h,\lambda}=T_h^\lambda v_{n+1}^{h,\lambda},
    \qquad n=N-1,\ldots,0.
\end{equation}

\begin{remark}[CE-score randomization]\label{rem:CE-score}
The operator \(T_h^\lambda\) regularizes deterministic-action certainty
equivalent scores. It is not the same as first randomizing a quote and then
putting the outer risk-sensitive certainty equivalent around the mixture. This
is why continuous-time performance is stated under the fresh-sampling
relaxed-control model in \Cref{sec:fresh-sampling}.
\end{remark}

\section{Main Convergence Result}
\label{sec:main-convergence}

\begin{theorem}[Discrete soft Bellman limit]\label{thm:main}
Under Assumption~\ref{ass:model}, there exist constants \(C>0\), \(h_0>0\), and
\(\lambda_0>0\), independent of \(h\) and \(\lambda\), such that for every
\(h=T/N\le h_0\) and \(0<\lambda\le\lambda_0\),
\begin{equation}\label{eq:main-rate}
    \max_{0\le n\le N}\max_{q\in\cQ}
    \abs{v_{n,q}^{h,\lambda}-v_q^0(t_n)}
    \le
    C\left[h+\lambda\bigl(1+\abs{\log\lambda}\bigr)\right].
\end{equation}
Consequently,
\begin{equation}\label{eq:value-rate}
    \abs{x+qs+v_{n,q}^{h,\lambda}-V^*(t_n,x,s,q)}
    \le
    C\left[h+\lambda\bigl(1+\abs{\log\lambda}\bigr)\right].
\end{equation}
\end{theorem}

\paragraph{Proof strategy.}
The proof combines four estimates. First, a priori bounds place \(v^0\),
\(v^\lambda\), and \(v^{h,\lambda}\) in a common compact region. Second, hard
verification identifies \(V^*(t,x,s,q)=x+qs+v_q^0(t)\). Third, the log-sum-exp
Hamiltonian has entropy bias \(O(\lambda(1+|\log\lambda|))\), yielding the
soft-to-hard error by backward Gronwall. Fourth, the finite-state
Feynman--Kac formula gives
\(T_h^\lambda\varphi=\varphi+hH^\lambda(\varphi)+O(h^2)\); nonexpansiveness
then propagates this local error to \(O(h)\). Full details are in
\Cref{app:auxiliary,app:soft-hard,app:fk-consistency}.

\section{Fresh-Sampling Policy Performance}
\label{sec:fresh-sampling}

\subsection{Fresh-sampling relaxed systems}

A fresh-sampling relaxed Markov policy is a measurable family
\[
    \pi_{u,q}\in\cP(A),
    \qquad (u,q)\in[t,T]\times\cQ.
\]
This is the marked point-process counterpart of the mixed-policy and
relaxed-control viewpoint used in recent continuous-time discretization limits
\cite{CarmonaLauriere2025,JiaOuyangZhang2025,PhamZhangZhu2026}.
The associated relaxed system is defined by two marked point measures
\(M^a(du,d\delta)\) and \(M^b(du,d\delta)\) on \([t,T]\times A\), with no common
jumps, whose predictable compensators are
\begin{equation}\label{eq:fresh-compensators}
\begin{aligned}
    \mu^a(du,d\delta)
    &=\one_{\{q_{u-}>-Q\}}\Lambda^a(\delta^a)\pi_{u,q_{u-}}(d\delta)\,du,\\
    \mu^b(du,d\delta)
    &=\one_{\{q_{u-}<Q\}}\Lambda^b(\delta^b)\pi_{u,q_{u-}}(d\delta)\,du.
\end{aligned}
\end{equation}
The inventory and cash dynamics are
\[
    q_u=q+M^b([t,u]\times A)-M^a([t,u]\times A),
\]
\[
    X_u=x+
    \int_t^u\int_A(S_r+\delta^a)M^a(dr,d\delta)
    -
    \int_t^u\int_A(S_r-\delta^b)M^b(dr,d\delta).
\]
A predictable-thinning realization of these compensators is given in
\Cref{lem:fresh-realization}.

The total accepted intensity is bounded by \(2\overline\Lambda\), so
\Cref{lem:exp-moments,lem:well-posedness} remain valid, with the total marked
count in place of \(N^a+N^b\).

For such a relaxed policy \(\pi\), define \(J(t,x,s,q;\pi)\) by the same
risk-sensitive certainty equivalent, evaluated under the marked fresh-sampling
system. The corresponding policy-evaluation identity is proved in
\Cref{lem:relaxed-verification}.

\begin{remark}[Fresh sampling versus static sampling]\label{rem:fresh-static}
The results in this section use the compensators \eqref{eq:fresh-compensators}.
They do not cover a scheme that samples one quote at the beginning of a time
interval and freezes it while keeping an outer risk-sensitive certainty
equivalent around the mixture.
\end{remark}

\subsection{Exact Bellman Gibbs policy}

The exact Bellman Gibbs policy is
\begin{equation}\label{eq:gibbs-exact}
    \pi_{n,q}^{\rm ex}(d\delta)
    :=
    \frac{
        \exp\left(
            C_h^\delta\bigl(v_{n+1}^{h,\lambda}\bigr)(q)/(h\lambda)
        \right)
    }{
        \int_A
        \exp\left(
            C_h^\zeta\bigl(v_{n+1}^{h,\lambda}\bigr)(q)/(h\lambda)
        \right)
        \nu(d\zeta)
    }
    \nu(d\delta),
\end{equation}
and is implemented as \(\pi_{t,q}^{\rm ex}:=\pi_{n,q}^{\rm ex}\) on
\([t_n,t_{n+1})\). We write \(\pi^{\rm ex}\) for this fresh-sampling policy.
The Borel measurability needed for the fresh-sampling implementation follows
from \Cref{lem:gibbs-measurable}.

\begin{corollary}[Fresh-sampling performance of the exact Bellman policy]\label{cor:performance-exact}
Under Assumption~\ref{ass:model}, there exist constants \(C>0\), \(h_0>0\), and
\(\lambda_0>0\) such that, for every \(h=T/N\le h_0\),
\(0<\lambda\le\lambda_0\), \(0\le n\le N-1\), \(t\in[t_n,t_{n+1})\), and
\((x,s,q)\in\R\times\R\times\cQ\), the fresh-sampling implementation satisfies
\[
    0\le V^*(t,x,s,q)-J(t,x,s,q;\pi^{\rm ex})
    \le
    C\left[h+\lambda(1+\abs{\log\lambda})\right].
\]
\end{corollary}

\paragraph{Proof idea.}
The result follows from the relaxed verification identity, the exact Gibbs
local regret bound, and a backward Gronwall comparison with the hard ODE; see
\Cref{app:policy-proofs}.

\subsection{Hamiltonian-Gibbs proxy policy}

The soft-HJB Euler approximation is
\begin{equation}\label{eq:soft-euler}
    \widehat v_{N,q}^{h,\lambda}=-\Phi q^2,
    \qquad
    \widehat v_n^{h,\lambda}
    =\widehat v_{n+1}^{h,\lambda}
     +hH^\lambda(\widehat v_{n+1}^{h,\lambda}).
\end{equation}
The Hamiltonian-Gibbs proxy is
\begin{equation}\label{eq:gibbs-ham}
    \pi_{n,q}^{\rm Ham}(d\delta)
    :=
    \frac{
        \exp\left(H_q(\widehat v_{n+1}^{h,\lambda},\delta)/\lambda\right)
    }{
        \int_A\exp\left(H_q(\widehat v_{n+1}^{h,\lambda},\zeta)/\lambda\right)
        \nu(d\zeta)
    }
    \nu(d\delta),
\end{equation}
and is implemented as \(\pi_{t,q}^{\rm Ham}:=\pi_{n,q}^{\rm Ham}\) on
\([t_n,t_{n+1})\). We write \(\pi^{\rm Ham}\) for this fresh-sampling policy.
The Borel measurability of this kernel is also covered by
\Cref{lem:gibbs-measurable}.

\begin{proposition}[Fresh-sampling performance of the Hamiltonian-Gibbs proxy]\label{prop:ham-performance}
Under Assumption~\ref{ass:model}, there exist constants \(C>0\), \(h_0>0\), and
\(\lambda_0>0\) such that, for every \(h=T/N\le h_0\),
\(0<\lambda\le\lambda_0\), \(0\le n\le N-1\), \(t\in[t_n,t_{n+1})\), and
\((x,s,q)\in\R\times\R\times\cQ\), the fresh-sampling implementation of
\eqref{eq:gibbs-ham} satisfies
\[
    0\le V^*(t,x,s,q)-J(t,x,s,q;\pi^{\rm Ham})
    \le
    C\left[h+\lambda(1+\abs{\log\lambda})\right].
\]
Thus the large-scale numerical policy based on the Hamiltonian-Gibbs density is
covered directly by the theory and need not be justified only through an
exact-Bellman proxy comparison.
\end{proposition}

\paragraph{Proof idea.}
Use the same ODE comparison as for \(\pi^{\rm ex}\), replacing the exact
Bellman local regret by the Hamiltonian-Gibbs local regret; see
\Cref{app:policy-proofs}.

\section{Quote Concentration}
\label{sec:quote-concentration}

For \(q\in\cQ\), define the active projection
\[
    P_q(\xi^a,\xi^b):=(\one_{\{q>-Q\}}\xi^a,\one_{\{q<Q\}}\xi^b).
\]
Let
\[
    \cA_q^*(t):=\argmax_{\delta\in A}H_q(v^0(t),\delta),
\]
and
\[
    \dist_q(\delta,\cA_q^*(t))
    :=
    \inf_{\delta^*\in\cA_q^*(t)}\abs{P_q(\delta-\delta^*)}.
\]

\begin{assumption}[Active-coordinate quadratic growth]\label{ass:qg}
There exists \(\mu>0\) such that, for all \(t\in[0,T]\), \(q\in\cQ\), and
\(\delta\in A\),
\[
    H_q^0(v^0(t))-H_q(v^0(t),\delta)
    \ge
    \frac\mu2\dist_q^2(\delta,\cA_q^*(t)).
\]
\end{assumption}

\begin{lemma}[Exponential-intensity curvature certificate]\label{lem:exp-curvature}
Assume
\[
    \Lambda^a(\delta)=\alpha_a e^{-k_a\delta},
    \qquad
    \Lambda^b(\delta)=\alpha_b e^{-k_b\delta},
\]
with \(\alpha_a,\alpha_b,k_a,k_b>0\). Define
\[
    D_a:=\overline\delta+
    \sup_{t\in[0,T],\ q>-Q}\{v_{q-1}^0(t)-v_q^0(t)\},
\]
\[
    D_b:=\overline\delta+
    \sup_{t\in[0,T],\ q<Q}\{v_{q+1}^0(t)-v_q^0(t)\},
\]
and
\[
    \Theta_a:=\frac2\gamma\log\left(1+\frac\gamma{k_a}\right),
    \qquad
    \Theta_b:=\frac2\gamma\log\left(1+\frac\gamma{k_b}\right).
\]
If
\[
    D_a<\Theta_a,
    \qquad
    D_b<\Theta_b,
\]
then Assumption~\ref{ass:qg} holds. More precisely, set
\[
    \mu_a:=
    \frac{\alpha_a}{\gamma}e^{-k_a\overline\delta}
    \left[(k_a+\gamma)^2e^{-\gamma D_a}-k_a^2\right]>0,
\]
\[
    \mu_b:=
    \frac{\alpha_b}{\gamma}e^{-k_b\overline\delta}
    \left[(k_b+\gamma)^2e^{-\gamma D_b}-k_b^2\right]>0.
\]
Then Assumption~\ref{ass:qg} holds with \(\mu=\min\{\mu_a,\mu_b\}\).
\end{lemma}

\paragraph{Proof idea.}
The derivative calculation showing active-coordinate strong concavity is given
in \Cref{app:quote-proofs}.

\begin{corollary}[Quote concentration for exact and Hamiltonian-Gibbs policies]\label{cor:quote-concentration}
Suppose that Assumption~\ref{ass:qg} holds. There exist constants \(C>0\),
\(h_0>0\), and
\(\lambda_0>0\) such that, for every \(h=T/N\le h_0\),
\(0<\lambda\le\lambda_0\), \(0\le n\le N-1\), and \(t\in[t_n,t_{n+1})\), the
following holds. Let \(\pi\) be either the exact Bellman Gibbs policy
\(\pi^{\rm ex}\) or the Hamiltonian-Gibbs proxy \(\pi^{\rm Ham}\), and write
\(\pi_{n,q}\) for the corresponding discrete kernel. Then, for every
\(q\in\cQ\),
\[
    \int_A\dist_q^2(\delta,\cA_q^*(t))\pi_{n,q}(d\delta)
    \le
    C\left[h+\lambda(1+\abs{\log\lambda})\right].
\]
If the active optimizer is unique, choose
\(\delta_q^*(t)\in\cA_q^*(t)\) so that \(P_q\delta_q^*(t)\) is the unique active
optimizer. The map \(t\mapsto P_q\delta_q^*(t)\) is Borel measurable because
\(t\mapsto v^0(t)\) is continuous and
\((t,\delta)\mapsto H_q(v^0(t),\delta)\) is continuous on the compact set
\(A\). On inactive coordinates we fix the lower endpoint
\(\underline\delta\), which gives a Borel selector. Write
\[
    \bar\delta_{n,q}^{\pi}:=\int_A\delta\,\pi_{n,q}(d\delta),
    \qquad
    \bar\delta_{h,\lambda}^{\pi}(t,q):=\bar\delta_{n,q}^{\pi}
    \quad\text{for }t\in[t_n,t_{n+1}).
\]
Then
\[
    \abs{P_q(\bar\delta_{n,q}^{\pi}-\delta_q^*(t))}^2
    \le
    C\left[h+\lambda(1+\abs{\log\lambda})\right],
\]
and
\[
    \sum_{q\in\cQ}\int_0^T
    \abs{P_q(\bar\delta_{h,\lambda}^{\pi}(t,q)-\delta_q^*(t))}^2dt
    \le
    C\left[h+\lambda(1+\abs{\log\lambda})\right].
\]
\end{corollary}

\paragraph{Proof idea.}
Combine local regret with active-coordinate quadratic growth, and then apply
Jensen's inequality for the mean-quote bound; see \Cref{app:quote-proofs}.

\section{Numerical Experiments}
\label{sec:numerics}

This section validates the quantitative bounds in a controlled finite-inventory
Avellaneda--Stoikov market-making model
\cite{AvellanedaStoikov2008,GueantLehalleFernandezTapia2013}. The experiments
are designed to test four claims: value convergence, fresh-sampling policy
performance, active-coordinate quote concentration, and the numerical
consistency between the exact Bellman policy and its Hamiltonian-Gibbs
computational proxy.

\subsection{Experimental setup}
\label{subsec:numerical-setup}

We use exponential arrival intensities
\[
    \Lambda^a(\delta)=\alpha_a e^{-k_a\delta},
    \qquad
    \Lambda^b(\delta)=\alpha_b e^{-k_b\delta}.
\]
Unless otherwise stated, the parameters are
\[
    T=1,\quad Q=5,\quad \sigma=0.20,\quad
    \gamma=0.10,\quad
    \Phi=0.02,\quad
    \eta=0.005,
\]
\[
    A=[0.01,0.70]^2,
    \qquad
    \alpha_a=\alpha_b=1.50,
    \qquad
    k_a=k_b=1.50.
\]
The reference action measure is the uniform probability measure on \(A\).
The hard continuous-time value \(v^0\) is computed by a fine-grid RK4 solver
and is used as the numerical ground truth. Unless otherwise stated, the
reference ODE step size is \(h_{\rm ref}=10^{-3}\), and action integrals
are evaluated by Gauss--Legendre quadrature with 61 nodes. Exact Bellman
validation uses 17 Gauss--Legendre nodes per coordinate for the action
integral in \(T_h^\lambda\).

The large-scale experiments use the soft-HJB Euler scheme
\[
    \widehat v_{N,q}^{h,\lambda}=-\Phi q^2,
    \qquad
    \widehat v_n^{h,\lambda}
    =
    \widehat v_{n+1}^{h,\lambda}
    +hH^\lambda(\widehat v_{n+1}^{h,\lambda}),
\]
and its Hamiltonian-Gibbs policy
\[
    \pi_{n,q}^{\rm Ham}(d\delta)
    \propto
    \exp\left(
        \frac{H_q(\widehat v_{n+1}^{h,\lambda},\delta)}{\lambda}
    \right)\nu(d\delta).
\]
The exact Bellman Gibbs policy is
\[
    \pi_{n,q}^{\rm ex}(d\delta)
    \propto
    \exp\left(
        \frac{C_h^\delta\bigl(v_{n+1}^{h,\lambda}\bigr)(q)}{h\lambda}
    \right)\nu(d\delta),
\]
where \(v^{h,\lambda}\) denotes the exact CE-score Bellman recursion. Thus
\(\widehat v^{h,\lambda}\) is used for the large-scale Hamiltonian-Gibbs
figures, while \(v^{h,\lambda}\) is used only in the exact Bellman validation.
The theoretical connection is the one-step consistency expansion
\[
    T_h^\lambda\varphi
    =
    \varphi+hH^\lambda(\varphi)+O(h^2).
\]
We verify this exact-vs-proxy bridge numerically in \Cref{subsec:exact-validation}.

All log-log slopes reported below are least-squares slopes on the displayed
finite grids. When slope-sensitivity intervals are produced by the code, they
are obtained by resampling deterministic grid points and should not be
interpreted as statistical confidence intervals. Throughout the numerical
section, unless otherwise stated,
\[
    CE(\pi):=J(0,0,0,0;\pi).
\]
When scalar policy gaps are reported, we write \(V^*:=V^*(0,0,0,0)\).

\begin{table}[t]
\centering
\caption{Baseline parameter values used in the numerical experiments.}
\label{tab:numerical-params}
\begin{tabular}{ll}
\toprule
Parameter & Value \\
\midrule
Time horizon & \(T=1\) \\
Inventory bound & \(Q=5\) \\
Initial inventory & \(q_0=0\) \\
Midprice volatility & \(\sigma=0.20\) \\
Risk aversion & \(\gamma=0.10\) \\
Terminal inventory penalty & \(\Phi=0.02\) \\
Running inventory penalty & \(\eta=0.005\) \\
Quote domain & \(A=[0.01,0.70]^2\) \\
Arrival intensities & \(\Lambda^{a,b}(\delta)=1.50e^{-1.50\delta}\) \\
Reference action law & Uniform on \(A\) \\
Reference ODE step & \(h_{\rm ref}=10^{-3}\) \\
Action quadrature & 61 Gauss--Legendre nodes \\
Exact Bellman action quadrature & 17 Gauss--Legendre nodes \\
\bottomrule
\end{tabular}
\end{table}

\subsection{Value convergence}
\label{subsec:value-convergence}

In this section, \(a\lesssim b\) means that \(a\) is bounded above by a
structural constant times \(b\).
For the computational proxy, \Cref{lem:soft-euler-error} gives
\[
    \|\widehat v^{h,\lambda}-v^0\|_\infty
    \lesssim
    h+\lambda(1+|\log\lambda|).
\]
We test the discretization and entropy-bias components separately. First, for
fixed \(\lambda=0.005\), we measure
\[
    E_h:=\|\widehat v^{h,\lambda}-v^\lambda\|_\infty.
\]
Second, using the continuous soft equation, we measure
\[
    E_\lambda:=\|v^\lambda-v^0\|_\infty.
\]

The discretization error is first order. Over
\[
    h\in\{0.020,0.010,0.005,0.0025,0.00125\},
\]
the fitted slope is
\[
    \widehat \beta_h=1.0017.
\]
The errors are
\[
\begin{array}{c|ccccc}
h & 0.020 & 0.010 & 0.005 & 0.0025 & 0.00125 \\
\hline
E_h & 8.28\cdot10^{-4} & 4.13\cdot10^{-4}
& 2.06\cdot10^{-4} & 1.03\cdot10^{-4}
& 5.15\cdot10^{-5}.
\end{array}
\]
This confirms the \(O(h)\) component.

For the entropy bias, the theorem gives an upper bound rather than an exact
asymptotic equivalence. We therefore report the normalized ratio
\[
    \frac{E_\lambda}{\lambda(1+|\log\lambda|)}.
\]
This ratio remains bounded across the tested range, with
\[
    \max_\lambda
    \frac{\|v^\lambda-v^0\|_\infty}
    {\lambda(1+|\log\lambda|)}
    =
    0.8432.
\]
The finite-grid log-log slope against \(\lambda(1+|\log\lambda|)\) is
\[
    \widehat\beta_\lambda=0.8414,
    \qquad
    \widehat\beta_{\lambda,\mathrm{tail}}=0.8683.
\]
We do not interpret this slope as an exact asymptotic exponent. The bounded
normalized ratio is the relevant observation for the theoretical upper bound.

\begin{figure}[t]
\centering
\begin{minipage}{0.48\textwidth}
    \centering
    \includegraphics[width=\textwidth]{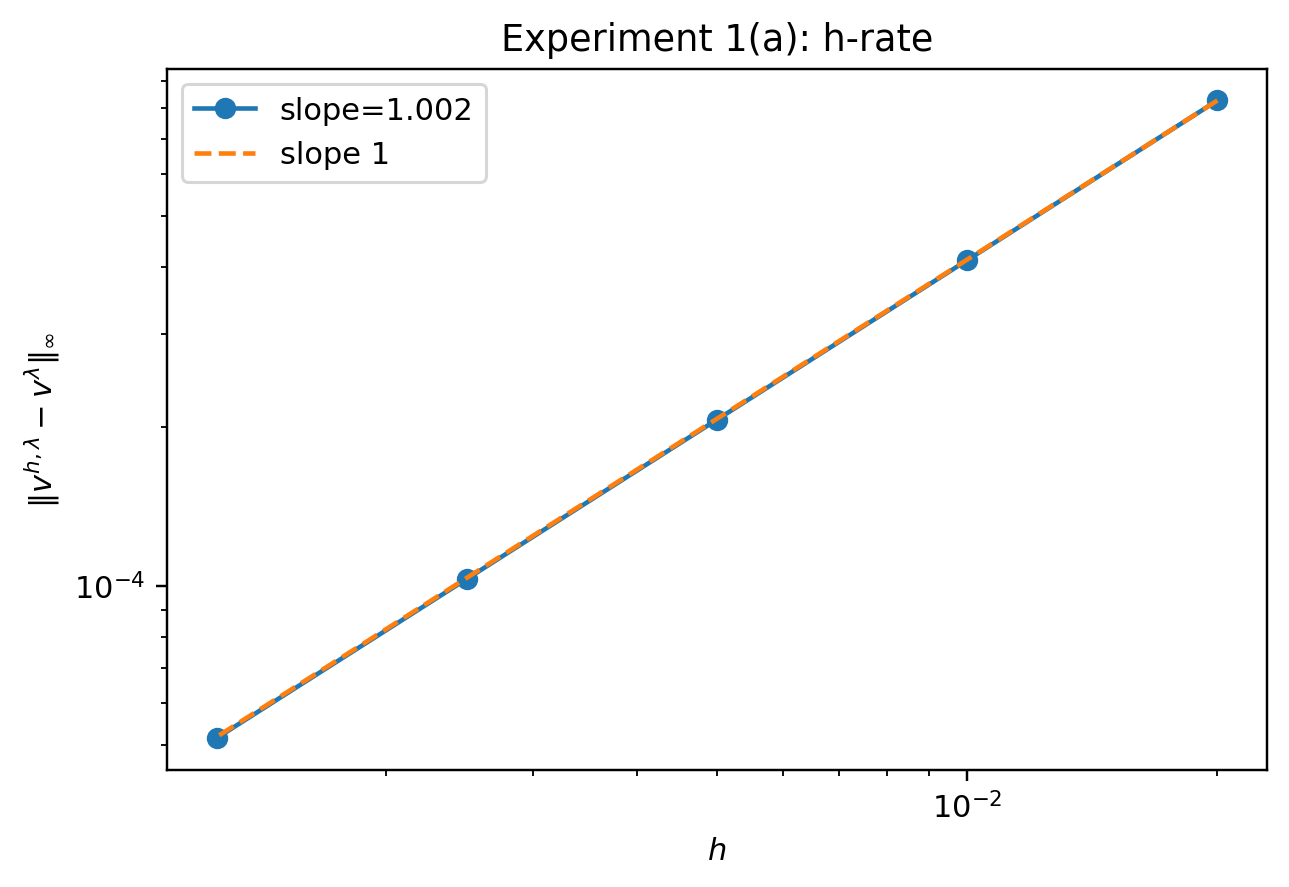}
\end{minipage}
\begin{minipage}{0.48\textwidth}
    \centering
    \includegraphics[width=\textwidth]{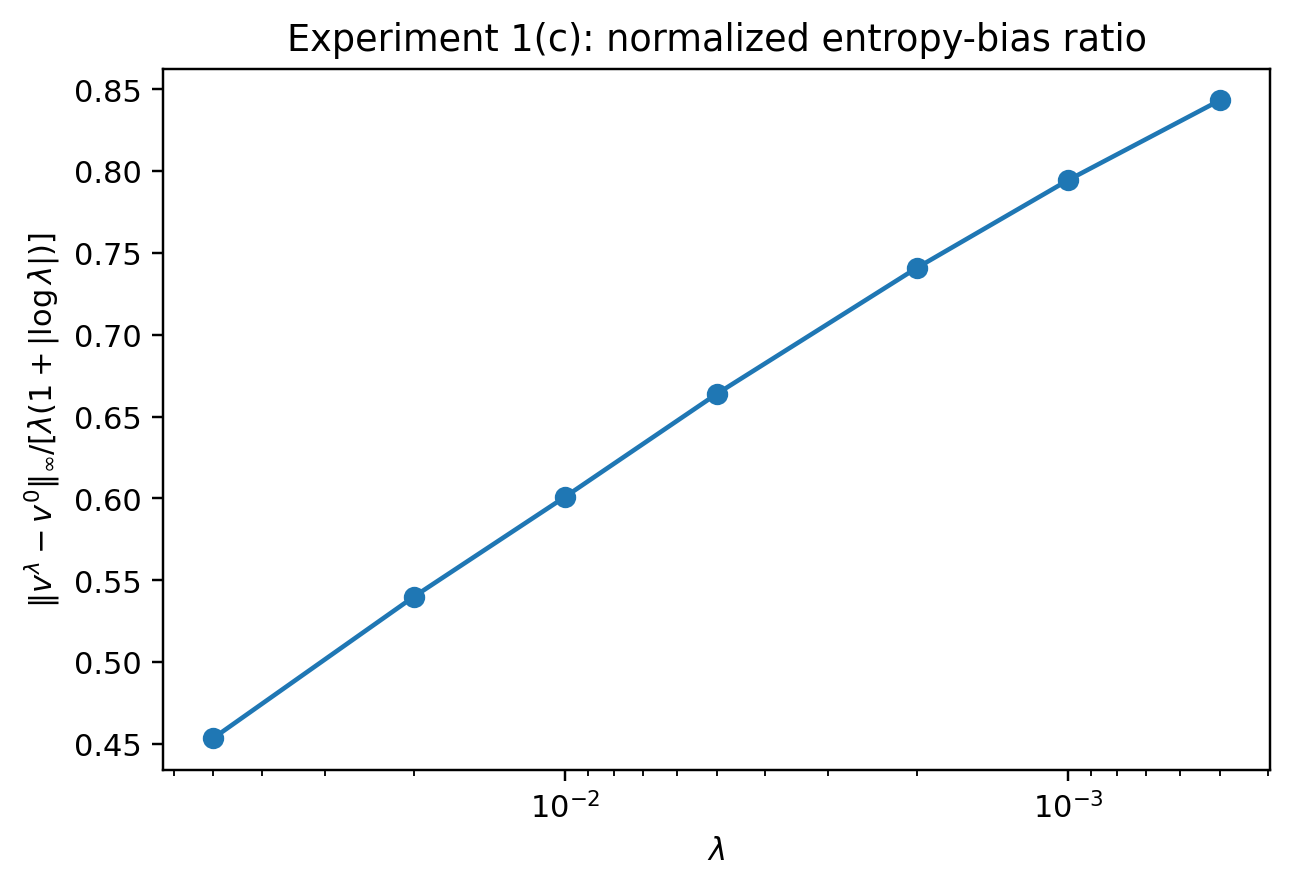}
\end{minipage}
\caption{Value convergence. Left: first-order discretization error
\(\|\widehat v^{h,\lambda}-v^\lambda\|_\infty\). Right: normalized entropy-bias ratio
\(\|v^\lambda-v^0\|_\infty/[\lambda(1+|\log\lambda|)]\), which remains bounded
on the tested grid.}
\label{fig:value-convergence}
\end{figure}

\subsection{Fresh-sampling policy performance}
\label{subsec:policy-performance}

We next test the policy-performance bound. The theoretical guarantee for the Hamiltonian-Gibbs proxy is
\[
    0\le
    V^*(t,x,s,q)-J(t,x,s,q;\pi^{\rm Ham})
    \lesssim
    h+\lambda(1+|\log\lambda|),
\]
when the randomized policy is implemented as a fresh-sampling relaxed Markov
control. Numerically, we evaluate the Hamiltonian-Gibbs proxy by solving the
corresponding fresh-sampling policy-evaluation ODE. The exact Bellman proxy gap
is validated separately in \Cref{subsec:exact-validation}.

Along the coupled path
\[
    (h,\lambda)\in
    \{(0.020,0.050),(0.010,0.020),(0.005,0.010),
    (0.0025,0.005),(0.00125,0.002)\},
\]
the certainty-equivalent performance gap decreases with fitted slope
\[
    \widehat\beta_{\mathrm{CE}}=1.2093.
\]
The observed gaps are reported in \Cref{tab:policy-gap}.

\begin{table}[t]
\centering
\caption{Fresh-sampling policy performance gap along the coupled path.}
\label{tab:policy-gap}
\begin{tabular}{ccccc}
\toprule
\(h\) & \(\lambda\) & \(h+\lambda(1+|\log\lambda|)\)
& \(V^*-CE(\pi^{\rm Ham})\) & fitted slope \\
\midrule
0.02000 & 0.050 & 0.219787 & 0.038608 & 1.2093 \\
0.01000 & 0.020 & 0.108240 & 0.015652 & 1.2093 \\
0.00500 & 0.010 & 0.061052 & 0.007809 & 1.2093 \\
0.00250 & 0.005 & 0.033992 & 0.003891 & 1.2093 \\
0.00125 & 0.002 & 0.015679 & 0.001577 & 1.2093 \\
\bottomrule
\end{tabular}
\end{table}

\begin{figure}[t]
\centering
\includegraphics[width=0.58\textwidth]{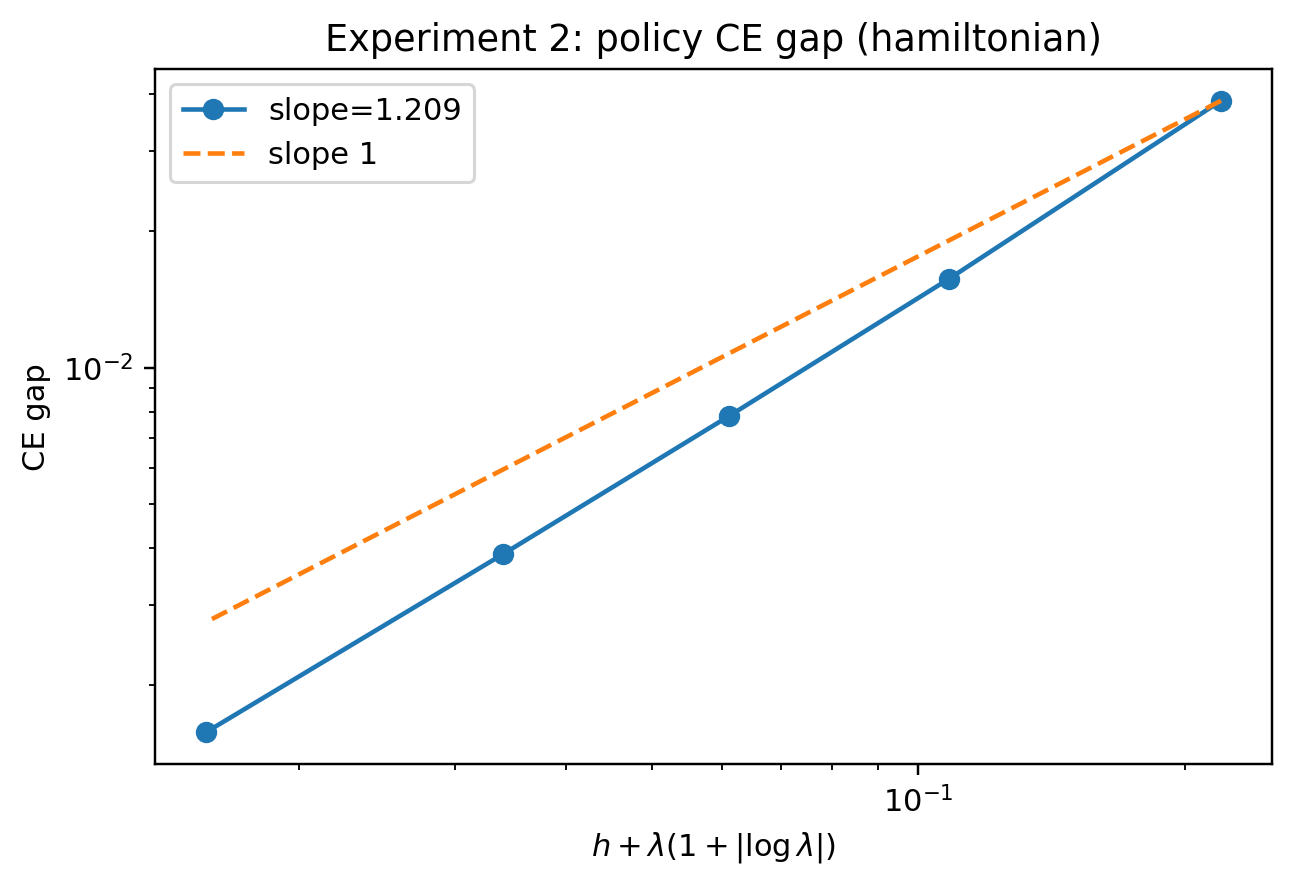}
\caption{Fresh-sampling policy performance gap. The certainty-equivalent gap
\(V^*-CE(\pi^{\rm Ham})\) decays approximately linearly in the theoretical
scale \(h+\lambda(1+|\log\lambda|)\).}
\label{fig:policy-gap}
\end{figure}

We also simulate the induced strategies under a fresh-sampling marked-Poisson
implementation with 5000 paths and common random numbers. The certainty
equivalent is the main metric because it is the risk-sensitive objective
controlled by the theorem. Mean reward, terminal PnL, and inventory diagnostics
are reported as secondary financial diagnostics in
\Cref{tab:financial-diagnostics}. In
the reported parameter regime, the Gibbs proxy has a certainty equivalent close
to the hard optimal policy and
substantially outperforms the constant-spread and inventory-linear baselines.

\subsection{Active-coordinate quote convergence}
\label{subsec:quote-convergence}

Let \(\bar\delta_{\rm Ham}^{h,\lambda}(t,q)\) denote the mean quote of the
Hamiltonian-Gibbs policy at time \(t\) and inventory \(q\), and let
\(\delta_q^*(t)\) denote the hard optimal quote used in the numerical
comparison. The quote convergence theorem controls the squared
active-coordinate error
\[
    \sum_{q\in\cQ}\int_0^T
    |P_q(\bar\delta_{\rm Ham}^{h,\lambda}(t,q)-\delta_q^*(t))|^2\,dt,
\]
not the ordinary \(L^2\)-norm itself. Therefore a first-order slope is expected
for the squared error, while the ordinary \(L^2\)-norm would have the square-root
rate.

The active-coordinate squared quote error decreases with fitted slope
\[
    \widehat\beta_{\mathrm{quote}^2}=1.4627.
\]
The integrated Hamiltonian regret
\[
    \sum_{q\in\cQ}\int_0^T
    \left[
        H_q^0(v^0(t))-
        \int_A H_q(v^0(t),\delta)\pi_{t,q}^{\rm Ham}(d\delta)
    \right]dt
\]
decreases with fitted slope
\[
    \widehat\beta_{\mathrm{regret}}=1.0853.
\]
The corresponding data are reported in \Cref{tab:quote-results}.

The sufficient curvature condition for exponential intensities is also
verified numerically. We obtain
\[
    D_a=D_b=1.0654,
    \qquad
    \frac{2}{\gamma}\log\left(1+\frac{\gamma}{k_a}\right)
    =
    \frac{2}{\gamma}\log\left(1+\frac{\gamma}{k_b}\right)
    =
    1.2908.
\]
Thus \(D_a,D_b\) are strictly below the theoretical thresholds. The resulting
sufficient curvature constants are
\[
    \mu_a=\mu_b=0.2692>0.
\]
Hence the active-coordinate quadratic growth condition holds for the baseline
parameter set.

\begin{figure}[t]
\centering
\begin{minipage}{0.48\textwidth}
    \centering
    \includegraphics[width=\textwidth]{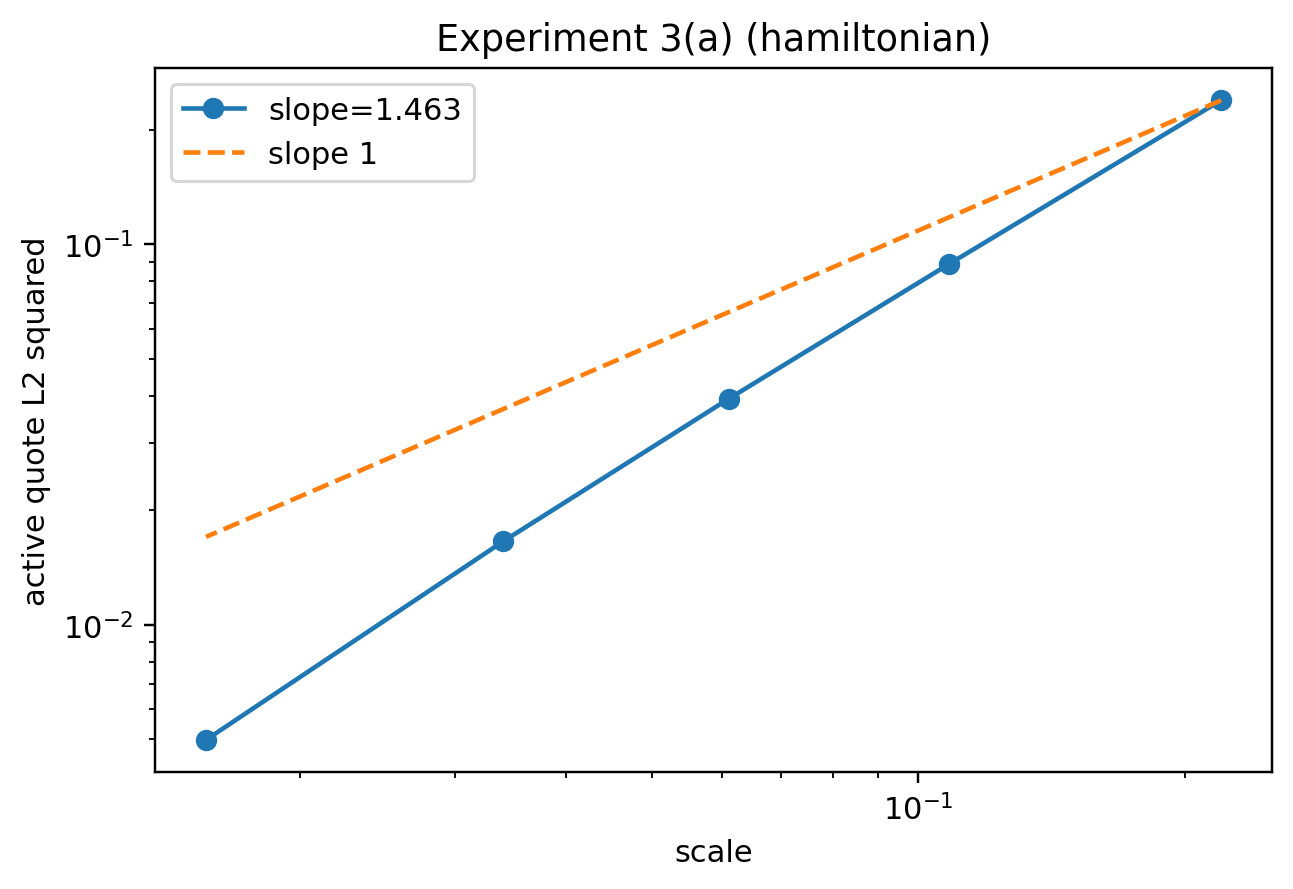}
\end{minipage}
\begin{minipage}{0.48\textwidth}
    \centering
    \includegraphics[width=\textwidth]{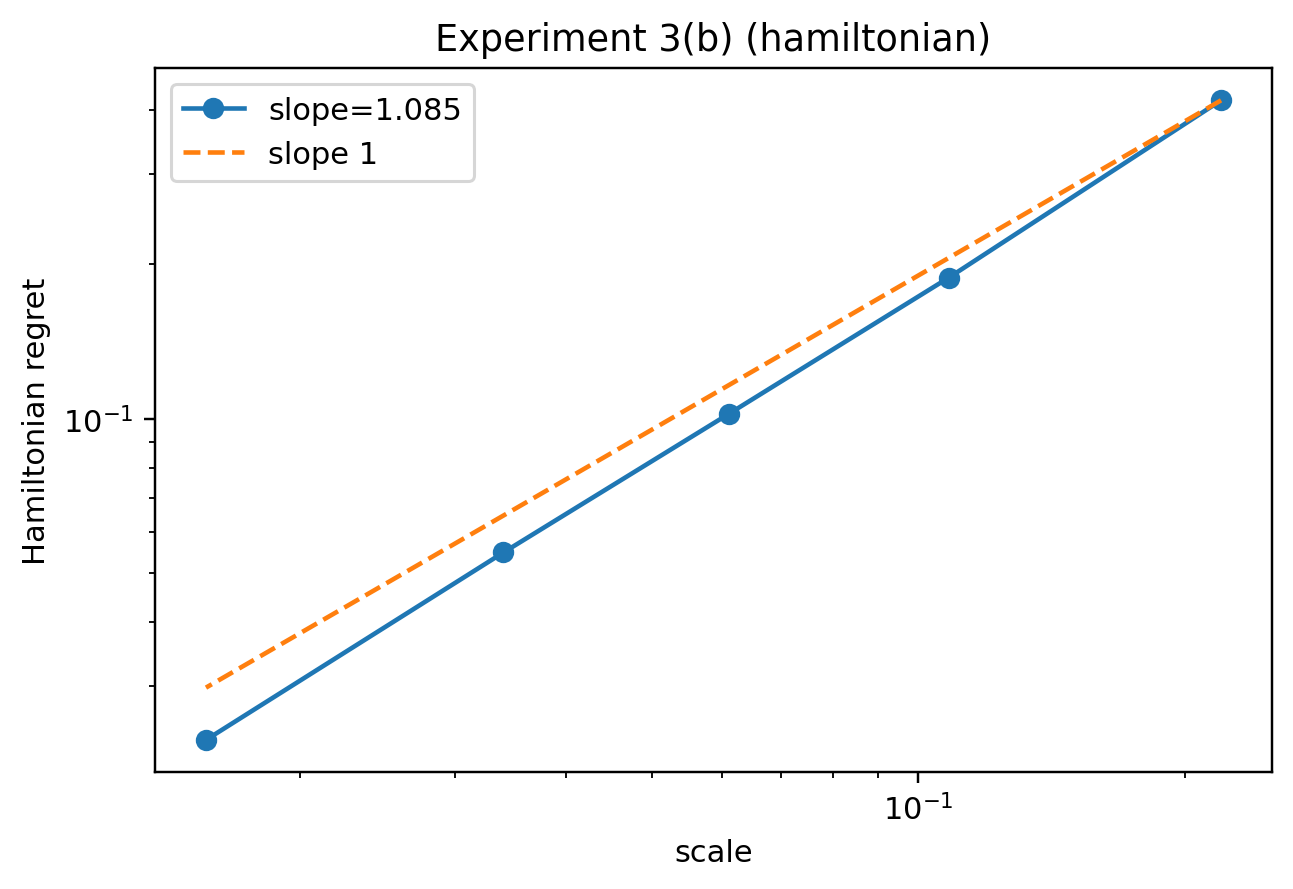}
\end{minipage}
\caption{Quote convergence. Left: squared active-coordinate quote error.
Right: integrated Hamiltonian regret. Both decay with the theoretical scale
\(h+\lambda(1+|\log\lambda|)\).}
\label{fig:quote-convergence}
\end{figure}

\subsection{Exact Bellman validation}
\label{subsec:exact-validation}

The exact Bellman operator is more expensive than the soft-HJB proxy because it
requires evaluating \(C_h^\delta\) for a grid of actions. We therefore use it
for validation rather than for all large-scale figures.

First, we directly test the one-step exact Bellman consistency:
\[
    \left\|
    T_h^\lambda\varphi-
    \bigl[\varphi+hH^\lambda(\varphi)\bigr]
    \right\|_\infty
    =
    O(h^2).
\]
For
\[
    h\in\{0.050,0.025,0.0125,0.00625\},
    \qquad
    \lambda=0.02,
\]
the fitted slope is
\[
    \widehat\beta_{\rm cons}=1.9873.
\]
Moreover, the normalized ratio \(\mathrm{error}/h^2\) stays nearly constant,
between \(0.0887\) and \(0.0911\). The raw values are reported in
\Cref{tab:exact-consistency}.

Second, we compare the exact Bellman Gibbs policy with the Hamiltonian-Gibbs
proxy on the same small grid:
\[
    h\in\{0.050,0.025,0.0125,0.00625\},
    \qquad
    \lambda=0.02.
\]
Here \(v^{h,\lambda}\) is the exact CE-score Bellman value, and
\(\widehat v^{h,\lambda}\) is the soft-HJB Euler proxy.
The corresponding mean-quote functions are denoted by
\(\bar\delta_{\rm ex}^{h,\lambda}(t,q)\) and
\(\bar\delta_{\rm Ham}^{h,\lambda}(t,q)\).
We also write
\[
    CE_{\rm ex}:=CE(\pi^{\rm ex}),
    \qquad
    CE_{\rm Ham}:=CE(\pi^{\rm Ham}).
\]
The maximum exact-vs-proxy discrepancies are
\[
    \max_h \|v^{h,\lambda}-\widehat v^{h,\lambda}\|_\infty
    =
    7.13\cdot10^{-4},
\]
\[
    \max_h
    \sum_{q\in\cQ}\int_0^T
    \left|
        P_q\left(
            \bar\delta_{\rm ex}^{h,\lambda}(t,q)
            -\bar\delta_{\rm Ham}^{h,\lambda}(t,q)
        \right)
    \right|^2\,dt
    =
    1.71\cdot10^{-6},
\]
and
\[
    \max_h
    |CE_{\rm ex}-CE_{\rm Ham}|
    =
    1.07\cdot10^{-5}.
\]
The fitted decay slopes are
\[
    \widehat\beta_{\rm value}=1.0154,
    \qquad
    \widehat\beta_{\rm quote}=2.1205,
    \qquad
    \widehat\beta_{\rm CE}=0.9836.
\]
Thus the exact-vs-proxy gap vanishes as \(h\to0\), supporting the use of the
Hamiltonian-Gibbs proxy in the large-scale experiments. The raw values are
reported in \Cref{tab:exact-vs-proxy}.

\begin{figure}[t]
\centering
\begin{minipage}{0.48\textwidth}
    \centering
    \includegraphics[width=\textwidth]{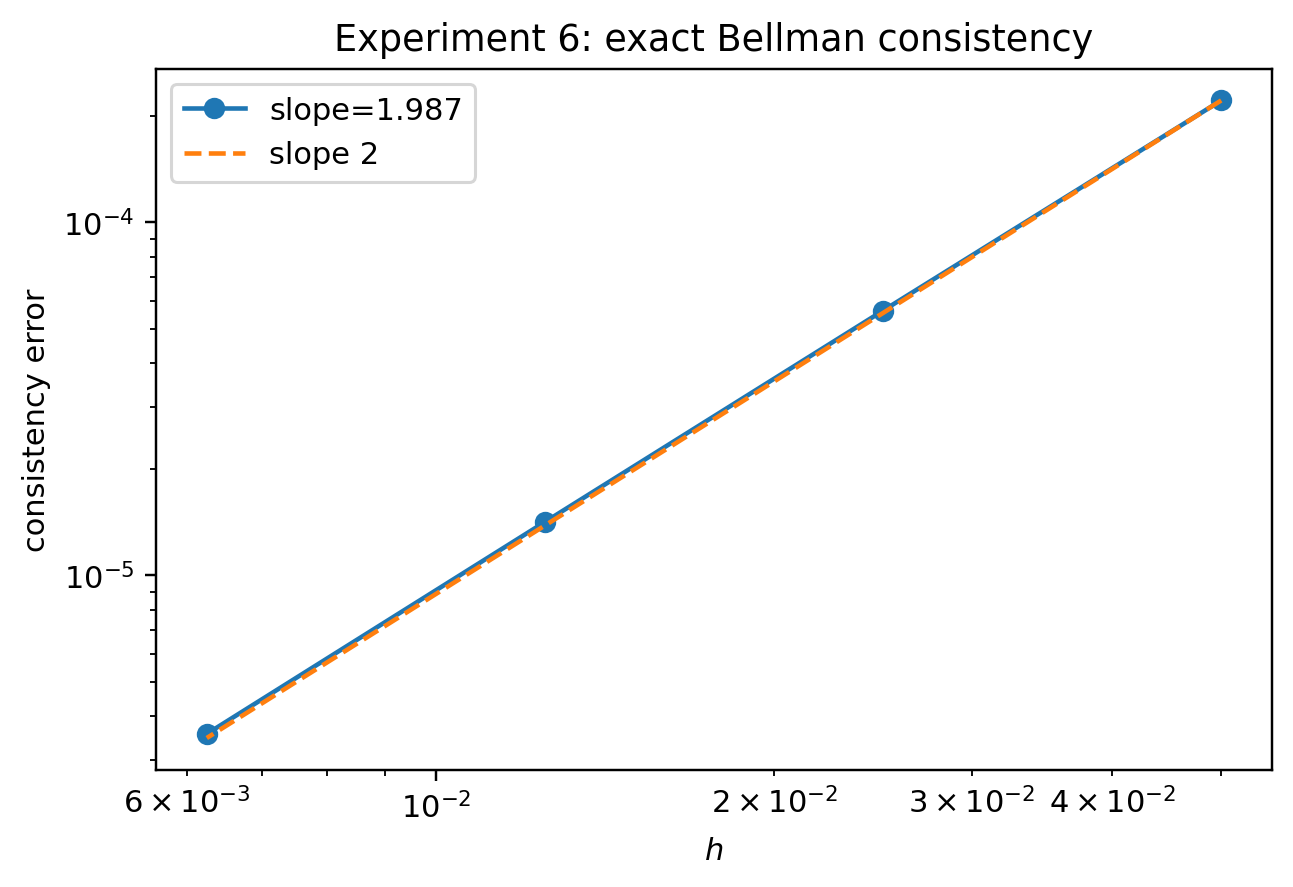}
\end{minipage}
\begin{minipage}{0.48\textwidth}
    \centering
    \includegraphics[width=\textwidth]{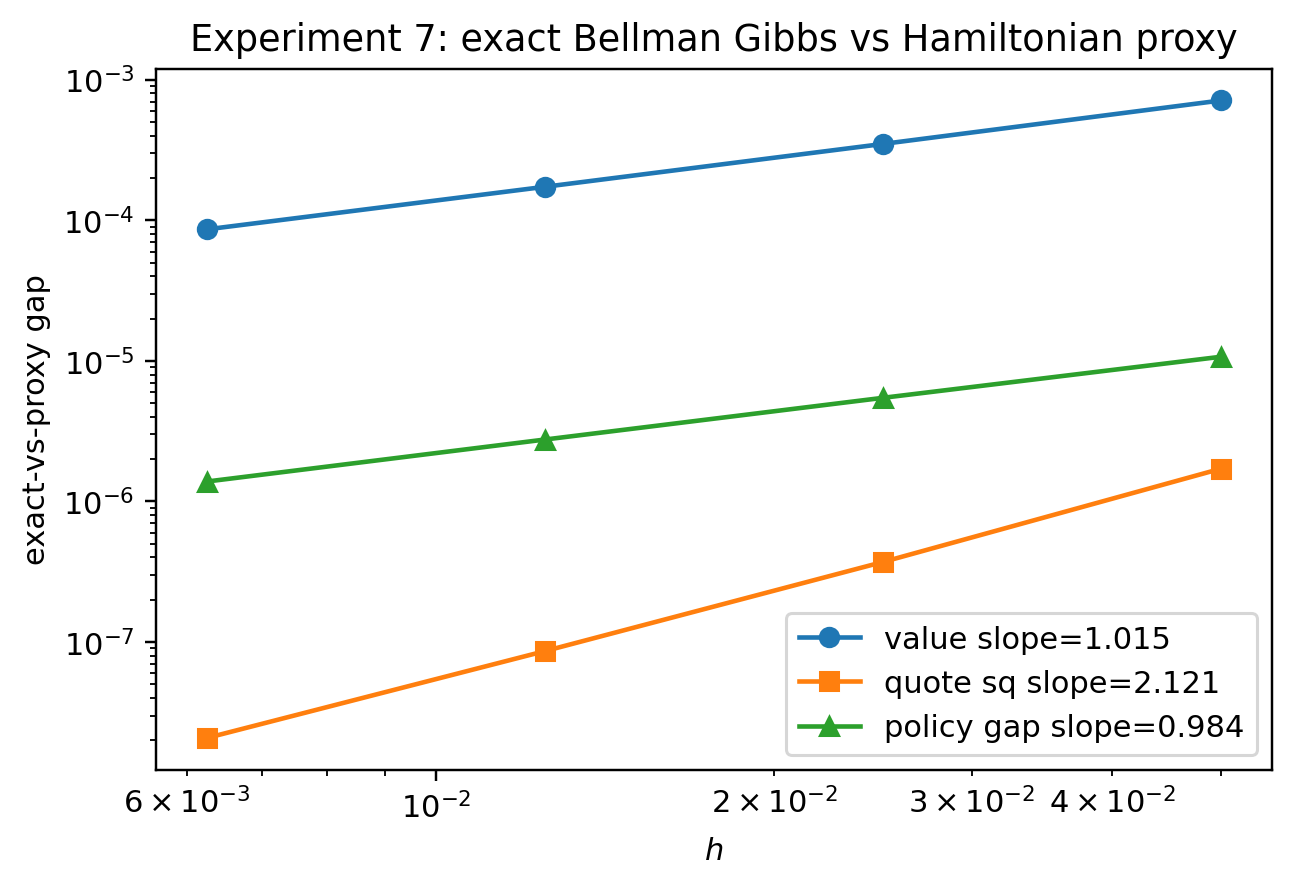}
\end{minipage}
\caption{Exact Bellman validation. Left: one-step consistency
\(\|T_h^\lambda\varphi-[\varphi+hH^\lambda(\varphi)]\|_\infty=O(h^2)\).
Right: exact Bellman Gibbs policy versus Hamiltonian-Gibbs proxy.}
\label{fig:exact-validation}
\end{figure}

Additional quadrature and risk-aversion robustness checks are reported in
\Cref{app:numerical-details}.

\subsection{Summary of numerical findings}
\label{subsec:numerical-summary}

The numerical evidence is consistent with the theoretical rates: first-order
discretization, bounded normalized entropy bias, decreasing policy and quote
gaps, and \(O(h^2)\) exact Bellman one-step consistency; see
\Cref{tab:numerical-summary}.

\begin{table}[t]
\centering
\caption{Summary of empirical rates and validation checks.}
\label{tab:numerical-summary}
\begin{tabular}{ll}
\toprule
Quantity & Empirical result \\
\midrule
Discretization error \(\|\widehat v^{h,\lambda}-v^\lambda\|_\infty\) & slope \(1.0017\) \\
Entropy bias \(\|v^\lambda-v^0\|_\infty\) & bounded normalized ratio; max \(0.8432\) \\
Policy CE gap & slope \(1.2093\) \\
Squared active quote error & slope \(1.4627\) \\
Integrated Hamiltonian regret & slope \(1.0853\) \\
Quadrature error & maximum \(8.20\cdot10^{-13}\) \\
Exact Bellman consistency & slope \(1.9873\) \\
Exact-vs-proxy value gap & slope \(1.0154\) \\
Exact-vs-proxy CE gap & slope \(0.9836\) \\
Exact-vs-proxy quote gap & slope \(2.1205\) \\
\bottomrule
\end{tabular}
\end{table}

Overall, the experiments show that the exact discrete Bellman theory, the
soft-HJB computational proxy, and the fresh-sampling continuous-time policy
implementation are mutually consistent in the finite-inventory risk-sensitive
market-making model.

\section{Conclusion}
\label{sec:conclusion}

This paper establishes an entropy-regularized CE-score Bellman framework for a
finite-inventory risk-sensitive market-making problem. The central object is
the exact CE-score Bellman operator
\[
    (T_h^\lambda\varphi)_q
    =
    h\lambda
    \log\int_A
    \exp\left(\frac{C_h^\delta\varphi(q)}{h\lambda}\right)\nu(d\delta),
\]
which regularizes deterministic-action certainty-equivalent scores rather than
risk-neutral rewards. This distinction is essential because the exponential
certainty equivalent is nonlinear and does not commute with quote
randomization.

The main theorem proves that the exact discrete soft value converges uniformly
to the hard continuous-time risk-sensitive value at rate
\[
    \max_{0\le n\le N}\max_{q\in\cQ}
    |v_{n,q}^{h,\lambda}-v_q^0(t_n)|
    \le
    C\bigl[h+\lambda(1+|\log\lambda|)\bigr].
\]
The \(O(h)\) term comes from the one-step consistency of the CE-score Bellman
operator, while the \(\lambda(1+|\log\lambda|)\) term comes from the entropy
bias of the log-sum-exp Hamiltonian over the two-dimensional quote set. The
same scale also controls the hard-value approximation error after the
risk-sensitive verification step.

The policy results show that the induced Gibbs policies are not only value
approximations. Under a fresh-sampling relaxed-control implementation, both the
exact Bellman Gibbs policy and the Hamiltonian-Gibbs proxy achieve
certainty-equivalent performance gaps of order
\[
    h+\lambda(1+|\log\lambda|).
\]
This implementation condition is part of the theorem, not a technical detail:
sampling a quote once and freezing it under an outer risk-sensitive certainty
equivalent is a different operation. The fresh-sampling formulation aligns the
marked point-process compensators with the averaged Hamiltonian in the relaxed
policy-evaluation equation.

The quote-concentration result gives a direct interpretation of the policy
bound at the level of bid and ask spreads. Under active-coordinate quadratic
growth, the Gibbs policies concentrate around the hard optimal quote set in
mean squared active-coordinate distance. For exponential arrival intensities,
the paper gives a verifiable curvature certificate that ensures this condition
in the baseline Avellaneda--Stoikov specification. The numerical experiments
confirm the predicted value convergence, policy-performance behavior,
quote-concentration scale, and exact-vs-proxy consistency.

Several extensions are natural. First, one can calibrate the arrival
intensities from limit-order-book data and study how estimation error propagates
through the CE-score Bellman operator. Second, the finite-inventory model can be
extended to include adverse-selection signals, transient market impact, latency,
or multi-asset inventory constraints. Third, the Hamiltonian-Gibbs proxy can be
combined with online learning of unknown intensities, where the exploration
distribution is updated from data rather than fixed ex ante. Finally, the
fresh-sampling relaxed-control formulation provides a rigorous starting point
for continuous-time risk-sensitive reinforcement learning algorithms whose
randomized policies carry direct performance guarantees in point-process
market-making models.

\bibliographystyle{unsrtnat}
\bibliography{references}

\appendix
\section{Auxiliary Estimates for the Market-Making Model}
\label{app:auxiliary}

Throughout the appendices, \(C\) denotes a finite constant depending only on
\[
    T,Q,\gamma,\sigma,\overline\Lambda,
    \underline\delta,\overline\delta,\Phi,\eta,m_-,m_+,
\]
and on the \(C^1\)-bounds of \(\Lambda^a,\Lambda^b\) on
\([\underline\delta,\overline\delta]\). Its value may change from line to line,
but it is independent of \(h\) and \(\lambda\).

\begin{lemma}[Realization of Markov strict controls]\label{lem:markov-realization}
Let \(\vartheta:[t,T]\times\cQ\to A\) be Borel measurable. Then there exists a
strict admissible system with
\[
    \delta_u=\vartheta(u,q_{u-}).
\]
\end{lemma}

\begin{proof}
Use a Brownian motion \(W\) and two independent Poisson random measures
\(\mathcal N^a(du,dz)\), \(\mathcal N^b(du,dz)\) on
\([t,T]\times[0,\infty)\), each with intensity \(du\,dz\) and independent of
\(W\), and thin them by
\[
    N_u^a-N_t^a
    :=\int_t^u\int_0^\infty
    \one_{\{z\le\Lambda^a(\vartheta^a(r,q_{r-}))\one_{\{q_{r-}>-Q\}}\}}
    \mathcal N^a(dr,dz),
\]
\[
    N_u^b-N_t^b
    :=\int_t^u\int_0^\infty
    \one_{\{z\le\Lambda^b(\vartheta^b(r,q_{r-}))\one_{\{q_{r-}<Q\}}\}}
    \mathcal N^b(dr,dz).
\]
The accepted intensities are bounded by \(\overline\Lambda\), hence the
finite-state jump equation is pathwise unique and nonexplosive. Independence
and diffuse time intensity give no common ask-bid jumps, and the compensators
are the stated ones.
\end{proof}

\begin{lemma}[Bounded-intensity exponential moments]\label{lem:exp-moments}
Let \(L\) be a counting process with predictable intensity \(\ell_u\le
\bar\ell\). For every stopping time \(\tau\le T\) and every \(a\ge0\),
\[
    \E\left[e^{a(L_\tau-L_t)}\right]
    \le
    \exp\{(e^a-1)\bar\ell(T-t)\}.
\]
If \(\beta\) is predictable and \(|\beta_u|\le\bar\beta\), then for every
\(\theta\in\R\),
\[
    \E\exp\left(\theta\int_t^\tau\beta_u\,dW_u\right)
    \le
    \exp\left(\frac{\theta^2\bar\beta^2(T-t)}2\right),
\]
and consequently, for every \(a\ge0\),
\begin{equation}\label{eq:combined-exp-moment}
    \sup_{\tau\le T}
    \E\exp\left(
        a\left|\int_t^\tau\beta_u\,dW_u\right|+a(L_\tau-L_t)
    \right)<\infty.
\end{equation}
In particular, \eqref{eq:combined-exp-moment} holds for
\(L=N^a+N^b\) in every strict admissible system with \(\bar\ell=2\overline
\Lambda\).
\end{lemma}

\begin{proof}
For the counting process, the exponential compensator formula gives that
\[
    \exp\left(
        a(L_{u\wedge\tau}-L_t)
        -(e^a-1)\int_t^{u\wedge\tau}\ell_r\,dr
    \right)
\]
is a nonnegative local martingale, hence a supermartingale after localization.
Since \(\ell_r\le\bar\ell\), this gives the counting estimate. The Brownian
bound follows from the exponential martingale for
\(\int_t^\cdot\beta_u\,dW_u\) and
\(\int_t^\tau\beta_u^2du\le\bar\beta^2(T-t)\). Finally use
\(e^{a|x|}\le e^{ax}+e^{-ax}\) and H\"older's inequality to obtain
\eqref{eq:combined-exp-moment}.
\end{proof}

\begin{lemma}[Non-explosion and inventory constraint]\label{lem:well-posedness}
Every strict admissible system satisfies
\[
    q_u\in\cQ,
    \qquad
    \E\left[(N_T^a-N_t^a)+(N_T^b-N_t^b)\right]
    \le2\overline\Lambda(T-t).
\]
In particular, \((N_T^a-N_t^a)+(N_T^b-N_t^b)<\infty\) a.s.
\end{lemma}

\begin{proof}
The compensator of \(N^a+N^b\) is bounded by \(2\overline\Lambda\,du\). Taking
expectations gives the stated first-moment estimate. A finite expectation for
the total count implies non-explosion. The indicators
\(\one_{\{q_{u-}>-Q\}}\) and \(\one_{\{q_{u-}<Q\}}\) remove jumps that would leave
\(\cQ\).
\end{proof}

\begin{lemma}[Well-defined exponential criterion]\label{lem:J-well-defined}
For every strict admissible system \(\mathfrak S\in\mathfrak A_t\),
\[
    0<
    \E^{\mathfrak S}_{t,x,s,q}
    \exp\left(
    -\gamma\left[
        X_T+q_TS_T-\Phi q_T^2
        -\int_t^T\eta q_u^2\,du
    \right]
    \right)
    <\infty .
\]
Consequently \(J(t,x,s,q;\mathfrak S)\) is well-defined.
\end{lemma}

\begin{proof}
Let
\[
    R_T:=
    X_T+q_TS_T-\Phi q_T^2-\int_t^T\eta q_u^2\,du .
\]
By \eqref{eq:self-financing},
\[
    R_T
    =
    x+qs+\int_t^Tq_{u-}\,dS_u
    +\int_t^T\delta_u^a\,dN_u^a
    +\int_t^T\delta_u^b\,dN_u^b
    -\Phi q_T^2-\int_t^T\eta q_u^2\,du .
\]
Set \(D_A:=\max\{|\underline\delta|,|\overline\delta|\}\). Since
\(|q_u|\le Q\) and \(|\delta_u^{a,b}|\le D_A\),
\[
    e^{-\gamma R_T}
    \le
    C
    \exp\left(
        \gamma\left|\int_t^Tq_{u-}\,dS_u\right|
        +\gamma D_A\bigl[(N_T^a-N_t^a)+(N_T^b-N_t^b)\bigr]
    \right),
\]
where \(C<\infty\) depends on the fixed initial state and structural
parameters. The right-hand side has finite expectation by
\Cref{lem:exp-moments}. Positivity is immediate because the exponential random
variable is strictly positive.
\end{proof}

\subsection{A priori bounds}

\begin{lemma}[One-period certainty-equivalent bounds]\label{lem:one-step-bound}
Let \(\mathbf 1\) denote the vector of ones in \(\R^{2Q+1}\). There exists
\(K<\infty\) such that, for every \(c\in\R\), every \(q\in\cQ\), every
deterministic quote \(\delta\in A\), and every \(0<h\le1\),
\[
    c-Kh\le C_h^\delta(c\mathbf 1)(q)\le c+Kh.
\]
\end{lemma}

\begin{proof}
Constants factor out of the certainty equivalent, so
\[
    C_h^\delta(c\mathbf 1)(q)
    =
    c-\frac1\gamma\log\E_{0,0,q}^\delta[e^{-\gamma Y_h}],
\]
where, by \eqref{eq:self-financing} and \(x=s=0\),
\[
    Y_h
    =
    X_h+q_hS_h-\int_0^h\eta q_u^2du
    =
    \int_0^h q_{u-}\,dS_u+\delta^aN_h^a+\delta^bN_h^b
    -\int_0^h\eta q_u^2du .
\]
Set \(D_A:=\max\{\abs{\underline\delta},\abs{\overline\delta}\}\). Since
\(\abs{q_u}\le Q\), \(\abs{\delta^{a,b}}\le D_A\), and the fill intensities are
bounded, \Cref{lem:exp-moments} gives
\(\E e^{-\gamma Y_h}\le e^{Ch}\). Jensen's inequality gives
\[
    \log\E e^{-\gamma Y_h}\ge -\gamma\E[Y_h].
\]
The stochastic integral has mean zero and
\(\E(N_h^a+N_h^b)\le2\overline\Lambda h\), so \(\abs{\E[Y_h]}\le Ch\). Hence
\(\abs{\log\E e^{-\gamma Y_h}}\le Ch\), proving the claim.
\end{proof}

\begin{lemma}[Uniform boundedness]\label{lem:bounds}
There exists \(B<\infty\), independent of \(h\) and \(\lambda\), such that
\[
    \sup_{t\in[0,T],q\in\cQ}\abs{v_q^0(t)}
    +
    \sup_{\substack{t\in[0,T],q\in\cQ\\0<\lambda\le1}}\abs{v_q^\lambda(t)}
    +
    \sup_{\substack{0<h\le1,\ 0<\lambda\le1\\ Nh=T,\ 0\le n\le N,\ q\in\cQ}}
    \abs{v_{n,q}^{h,\lambda}}
    \le B.
\]
\end{lemma}

\begin{proof}
Write the continuous equations forward in \(\tau=T-t\):
\[
    u_q^0(\tau)=v_q^0(T-\tau),
    \qquad
    u_q^\lambda(\tau)=v_q^\lambda(T-\tau).
\]
Then \(\dot u^0=H^0(u^0)\) and \(\dot u^\lambda=H^\lambda(u^\lambda)\). For every
\(y\) and \(q\),
\[
    H_q^0(y)\le\frac{2\overline\Lambda}{\gamma},
    \qquad
    H_q^\lambda(y)\le H_q^0(y)\le\frac{2\overline\Lambda}{\gamma},
\]
because \(1-e^{-\gamma x}\le1\) and \(\nu\) is a probability measure. Thus the
upper Dini derivative of the running maximum is bounded above by
\(2\overline\Lambda/\gamma\).

If \(q\) is a minimizer of \(y\), then every active neighbor increment satisfies
\(y_{q\pm1}-y_q\ge0\). Hence, on active sides,
\[
    \Delta_q^a(y,\delta)\ge\underline\delta,
    \qquad
    \Delta_q^b(y,\delta)\ge\underline\delta.
\]
Therefore \(H_q(y,\delta)\ge-K_-\) for a structural constant \(K_-\), whenever
\(q\) is a minimizer of \(y\). Consequently
\[
    H_q^0(y)\ge-K_-,
    \qquad
    H_q^\lambda(y)\ge\inf_{\delta\in A}H_q(y,\delta)\ge-K_-.
\]
The lower Dini derivative of the running minimum is bounded below by \(-K_-\).
Since \(-\Phi Q^2\le v_q^0(T)=v_q^\lambda(T)\le0\), the continuous solutions are
uniformly bounded and cannot explode before time zero.

For the discrete recursion, let
\[
    M_n:=\max_qv_{n,q}^{h,\lambda},
    \qquad
    m_n:=\min_qv_{n,q}^{h,\lambda}.
\]
Monotonicity of the certainty equivalent and of the log-sum-exp operator gives
\[
    T_h^\lambda(m_{n+1}\mathbf 1)
    \le T_h^\lambda v_{n+1}^{h,\lambda}
    \le T_h^\lambda(M_{n+1}\mathbf 1).
\]
By \Cref{lem:one-step-bound},
\[
    m_{n+1}-Kh\le m_n\le M_n\le M_{n+1}+Kh.
\]
Since \(-\Phi Q^2\le v_{N,q}^{h,\lambda}\le0\), iteration over \(N=T/h\)
steps gives the uniform discrete bound.
\end{proof}

\begin{lemma}[Local Lipschitz estimates]\label{lem:lip}
Let \(B\) be as in \Cref{lem:bounds}. There exists \(L<\infty\) such that, for
all \(y,z\) with \(\norm{y}_\infty,\norm{z}_\infty\le B\),
\[
    \abs{H_q(y,\delta)-H_q(z,\delta)}
    \le L\norm{y-z}_\infty
\]
uniformly over \(q\in\cQ\) and \(\delta\in A\). Consequently,
\[
    \abs{H_q^0(y)-H_q^0(z)}\le L\norm{y-z}_\infty,
    \qquad
    \abs{H_q^\lambda(y)-H_q^\lambda(z)}\le L\norm{y-z}_\infty.
\]
Moreover, \(H^\lambda\) is uniformly bounded on \(\norm{y}_\infty\le B\),
uniformly in \(0<\lambda\le1\).
\end{lemma}

\begin{proof}
On \(\norm{y}_\infty,\norm{z}_\infty\le B\), all active exponential factors in
\(H_q\) are uniformly bounded, and each \(\Delta_q^{a,b}\) is affine in \(y\).
Hence \(H_q(\cdot,\delta)\) is uniformly Lipschitz. Taking suprema gives the
bound for \(H^0\). The log-sum-exp map is one-Lipschitz under uniform
perturbations of the integrand, giving the bound for \(H^\lambda\). Uniform
boundedness on the bounded region follows from compactness of \(A\).
\end{proof}

\begin{lemma}[Time regularity of the reduced values]\label{lem:time-lip}
There exists \(L_t<\infty\), independent of \(0<\lambda\le1\), such that
\[
    \abs{v_q^0(t)-v_q^0(s)}
    +
    \abs{v_q^\lambda(t)-v_q^\lambda(s)}
    \le
    L_t\abs{t-s},
    \qquad s,t\in[0,T],\ q\in\cQ .
\]
\end{lemma}

\begin{proof}
By \Cref{lem:bounds,lem:lip}, \(H^0(v^0(t))\) and
\(H^\lambda(v^\lambda(t))\) are uniformly bounded on \([0,T]\). The ODEs
\(-\dot v^0=H^0(v^0)\) and
\(-\dot v^\lambda=H^\lambda(v^\lambda)\) therefore give the claim by
integration.
\end{proof}

\section{Hard Verification and Soft-to-Hard Error}
\label{app:soft-hard}

\subsection{Hard verification}

\begin{lemma}[Hard verification]\label{lem:verification}
The hard value function satisfies
\[
    V^*(t,x,s,q)=x+qs+v_q^0(t).
\]
\end{lemma}

\begin{proof}
Let \(V(t,x,s,q):=x+qs+v_q^0(t)\). Fix an arbitrary strict admissible system
\(\mathfrak S\in\mathfrak A_t\). Define
\[
    M_u:=
    \exp\left(
        -\gamma\left[
            V(u,X_u,S_u,q_u)-\int_t^u\eta q_r^2\,dr
        \right]
    \right).
\]
The exponential It\^o formula for jump semimartingales, using the no-common-jump
convention, gives, for a local martingale \(\mathcal M\),
\[
    dM_u
    =
    -\gamma M_{u-}
    \left[
        \dot v_{q_{u-}}^0(u)+H_{q_{u-}}(v^0(u),\delta_u)
    \right]du
    +d\mathcal M_u .
\]
The Brownian correction contributes \(-\gamma\sigma^2q^2/2\) inside \(H\). At an
ask fill, the jump in \(V\) equals \(\delta^a+v_{q-1}^0(u)-v_q^0(u)\), and at a
bid fill it equals
\(\delta^b+v_{q+1}^0(u)-v_q^0(u)\). Boundary indicators remove nonexistent
neighbors.

Since \(-\dot v_q^0=H_q^0(v^0)=\sup_{\zeta\in A}H_q(v^0,\zeta)\),
\[
    \dot v_q^0(u)+H_q(v^0(u),\delta_u)\le0.
\]
Thus \(M\) is a positive local submartingale.

By \eqref{eq:self-financing} and the boundedness of \(v^0\), quotes, and
inventory, the same estimate holds for every stopping time \(\tau\le T\):
\[
    M_{\tau}
    \le
    C\exp\left(
        C\left|\int_t^{\tau}q_{u-}\,dS_u\right|
        +C\bigl[(N_{\tau}^a-N_t^a)+(N_{\tau}^b-N_t^b)\bigr]
    \right).
\]
By \Cref{lem:exp-moments}, there exists \(p>1\) such that
\[
    \sup_{\tau\le T}\E[M_\tau^p]<\infty .
\]
Hence \(M\) is of class \(D\). Since \(M\) is a positive local submartingale,
it is a true submartingale. Therefore \(\E[M_T]\ge M_t\). Since
\[
    V(T,X_T,S_T,q_T)=X_T+q_TS_T-\Phi q_T^2,
\]
we obtain
\[
    J(t,x,s,q;\mathfrak S)\le x+qs+v_q^0(t).
\]
Taking the supremum over \(\mathfrak A_t\) gives the upper bound.

For the reverse inequality, \(t\mapsto v^0(t)\) is continuous and
\(\delta\mapsto H_q(v^0(t),\delta)\) is continuous on compact \(A\). By the
measurable maximum theorem, for each \(q\) there is a Borel selector
\[
    \delta_q^*(t)\in\argmax_{\delta\in A}H_q(v^0(t),\delta).
\]
By \Cref{lem:markov-realization}, the feedback \(\delta_u^*=\delta_{q_{u-}}^*(u)\)
is admissible. For this system,
\[
    \dot v_{q_{u-}}^0(u)+H_{q_{u-}}(v^0(u),\delta_u^*)=0,
\]
so the same exponential process has zero drift and is a local martingale.
The class \(D\) estimate above applies to this process as well. Hence it is a
true martingale, and \(\E[M_T]=M_t\). Equality is attained, proving the identity.
\end{proof}

\subsection{Entropy bias and soft-to-hard convergence}

\begin{lemma}[Entropy bias]\label{lem:entropy-bias}
For \(0<\lambda\le1\) and \(\norm{y}_\infty\le B\),
\[
    0\le H_q^0(y)-H_q^\lambda(y)
    \le C\lambda(1+\abs{\log\lambda}),
    \qquad q\in\cQ.
\]
\end{lemma}

\begin{proof}
Since \(\nu\) is a probability measure, \(H_q^\lambda(y)\le H_q^0(y)\). Let
\(\delta^*\in\argmax_{\delta\in A}H_q(y,\delta)\). On \(\norm{y}_\infty\le B\),
\(\delta\mapsto H_q(y,\delta)\) is uniformly Lipschitz, so
\[
    H_q(y,\delta)\ge H_q^0(y)-L_\delta\abs{\delta-\delta^*}.
\]
For \(0<r<r_0\), the rectangle geometry and the lower density bound imply
\[
    \nu(B_r(\delta^*)\cap A)\ge c_A r^2
\]
uniformly over \(\delta^*\in A\), including edge and corner points. Therefore
\[
    \int_A e^{H_q(y,\delta)/\lambda}\nu(d\delta)
    \ge
    e^{(H_q^0(y)-L_\delta r)/\lambda}c_A r^2.
\]
Taking \(\lambda\log\) gives
\[
    H_q^0(y)-H_q^\lambda(y)
    \le L_\delta r-\lambda\log c_A-2\lambda\log r.
\]
Choose \(r=\lambda\) for small \(\lambda\); enlarge \(C\) for the remaining
bounded range.
\end{proof}

\begin{lemma}[Soft-to-hard error]\label{lem:soft-hard}
For \(0<\lambda\le1\),
\[
    \sup_{t\in[0,T]}\norm{v^\lambda(t)-v^0(t)}_\infty
    \le C\lambda(1+\abs{\log\lambda}).
\]
\end{lemma}

\begin{proof}
The integral equations give
\[
    v^\lambda(t)-v^0(t)
    =
    \int_t^T\{H^\lambda(v^\lambda(r))-H^0(v^0(r))\}\,dr.
\]
Using \Cref{lem:bounds,lem:lip,lem:entropy-bias},
\[
    \norm{v^\lambda(t)-v^0(t)}_\infty
    \le
    \int_t^T
    \left[
        L\norm{v^\lambda(r)-v^0(r)}_\infty
        +C\lambda(1+\abs{\log\lambda})
    \right]dr.
\]
Backward Gronwall proves the estimate.
\end{proof}

\section{Finite-State Feynman--Kac Formula and One-Step Consistency}
\label{app:fk-consistency}

\subsection{Finite-state Feynman--Kac and one-step consistency}

\begin{lemma}[Finite-state Feynman--Kac for a fixed quote]\label{lem:finite-state-fk}
Fix \(\delta\in A\) and \(\varphi\in\R^{2Q+1}\). Set
\[
    u_q:=e^{-\gamma\varphi_q}.
\]
Define the finite-dimensional matrix \(K^\delta\) by
\begin{equation}\label{eq:Kdelta}
\begin{aligned}
    (K^\delta z)_q
    :={}&
    \left(
        \frac{\gamma^2\sigma^2}{2}q^2+\gamma\eta q^2
        -\Lambda_q^a(\delta^a)-\Lambda_q^b(\delta^b)
    \right)z_q \\
    &+
    \Lambda_q^a(\delta^a)e^{-\gamma\delta^a}z_{q-1}
    +
    \Lambda_q^b(\delta^b)e^{-\gamma\delta^b}z_{q+1},
\end{aligned}
\end{equation}
where boundary terms with nonexistent neighbors are omitted. Then
\begin{equation}\label{eq:finite-state-FK}
    \E_{0,0,q}^{\delta}
    \left[
        \exp\left(
        -\gamma\left[
            F_\varphi(X_h,S_h,q_h)-\int_0^h\eta q_u^2\,du
        \right]
        \right)
    \right]
    =
    (e^{hK^\delta}u)_q.
\end{equation}
\end{lemma}

\begin{proof}
Under a frozen deterministic quote, \(q\) is a continuous-time Markov chain on
\(\cQ\) with transition rates
\[
    q\to q-1 \text{ at rate } \Lambda_q^a(\delta^a),
    \qquad
    q\to q+1 \text{ at rate } \Lambda_q^b(\delta^b).
\]
By \eqref{eq:self-financing}, starting from \(x=s=0\),
\[
    X_h+q_hS_h
    =
    \int_0^h q_{u-}\,dS_u+\delta^aN_h^a+\delta^bN_h^b.
\]
Condition on the inventory path. Since \(S_u=\sigma W_u\) and the Brownian
motion is independent of the fixed-quote jump construction,
\[
    \E\left[
        \exp\left(-\gamma\int_0^h q_{u-}\,dS_u\right)
        \middle|\,q_\cdot
    \right]
    =
    \exp\left(\frac{\gamma^2\sigma^2}{2}\int_0^h q_u^2\,du\right).
\]
Writing \(\E_q\) for expectation over this finite-state chain started from
state \(q\), the left side of \eqref{eq:finite-state-FK} equals
\[
    \E_q\left[
        \exp\left(\int_0^h
        \left[\frac{\gamma^2\sigma^2}{2}q_u^2+\gamma\eta q_u^2\right]du
        \right)
        e^{-\gamma\delta^aN_h^a-
          \gamma\delta^bN_h^b}
        u_{q_h}
    \right],
\]
where the remaining expectation is over the finite-state chain.

Let \(G_h(q)\) denote this last expectation. We show that
\(\partial_hG_h=K^\delta G_h\) and \(G_0=u\). Fix state \(q\) and write
\[
    r_q:=\frac{\gamma^2\sigma^2}{2}q^2+\gamma\eta q^2,
    \qquad
    \ell_q^a:=\Lambda_q^a(\delta^a),
    \qquad
    \ell_q^b:=\Lambda_q^b(\delta^b).
\]
For \(\varepsilon\downarrow0\), the Markov property after the first
\(\varepsilon\)-interval gives, uniformly in \(q\) and \(h\) on compact
horizons,
\[
\begin{aligned}
    G_{h+\varepsilon}(q)
    ={}&
    \{1+r_q\varepsilon+o(\varepsilon)\}
    \Big[
        \{1-(\ell_q^a+\ell_q^b)\varepsilon\}G_h(q)  \\
    &\qquad
        +\ell_q^a\varepsilon e^{-\gamma\delta^a}G_h(q-1)
        +\ell_q^b\varepsilon e^{-\gamma\delta^b}G_h(q+1)
    \Big]
    +o(\varepsilon),
\end{aligned}
\]
where missing boundary terms are omitted. Indeed, the probability of two or
more jumps in \([0,\varepsilon]\) is \(o(\varepsilon)\), and the running factor
is \(1+r_q\varepsilon+o(\varepsilon)\). Hence
\[
    \frac{G_{h+\varepsilon}(q)-G_h(q)}{\varepsilon}
    \to
    (K^\delta G_h)_q.
\]
The right-hand side is continuous in \(h\), since the state space is finite.
Thus \(G_h\) solves the finite-dimensional linear ODE
\(\partial_hG_h=K^\delta G_h\) with \(G_0=u\). Uniqueness of this ODE yields
\(G_h=e^{hK^\delta}u\).
\end{proof}

\begin{lemma}[One-step risk-sensitive consistency]\label{lem:consistency}
There exist \(h_0>0\) and \(C<\infty\) such that, for every \(0<h\le h_0\),
\[
    \sup_{\substack{\norm{\varphi}_\infty\le B\\ q\in\cQ,\ \delta\in A}}
    \left|
        C_h^\delta\varphi(q)
        -\varphi_q
        -hH_q(\varphi,\delta)
    \right|
    \le Ch^2.
\]
Consequently,
\[
    \norm{T_h^\lambda\varphi-[\varphi+hH^\lambda(\varphi)]}_\infty
    \le Ch^2.
\]
\end{lemma}

\begin{proof}
Let \(u_q=e^{-\gamma\varphi_q}\). By \Cref{lem:finite-state-fk},
\[
    \exp\{-\gamma C_h^\delta\varphi(q)\}=(e^{hK^\delta}u)_q.
\]
Since \(K^\delta\) has uniformly bounded entries and
\(e^{-\gamma B}\le u_q\le e^{\gamma B}\),
\[
    (e^{hK^\delta}u)_q
    =u_q+h(K^\delta u)_q+O(h^2)
\]
uniformly. A direct computation from \eqref{eq:Kdelta} gives
\[
    \frac{(K^\delta u)_q}{u_q}
    =
    \frac{\gamma^2\sigma^2}{2}q^2+
    \gamma\eta q^2
    +\Lambda_q^a(\delta^a)(e^{-\gamma\Delta_q^a(\varphi,\delta)}-1)
    +\Lambda_q^b(\delta^b)(e^{-\gamma\Delta_q^b(\varphi,\delta)}-1)
    =-\gamma H_q(\varphi,\delta).
\]
Hence
\[
    (e^{hK^\delta}u)_q
    =u_q\left[1-\gamma hH_q(\varphi,\delta)+O(h^2)\right].
\]
For \(h\le h_0\), the bracket is uniformly positive. Taking
\(-\gamma^{-1}\log\) and using \(\log(1+x)=x+O(x^2)\) uniformly for small
\(x\) proves the displayed bound for \(C_h^\delta\varphi(q)\).

Write
\[
    C_h^\delta\varphi(q)
    =\varphi_q+h\{H_q(\varphi,\delta)+\varepsilon_h(\delta)\},
    \qquad
    \sup_{\delta\in A}\abs{\varepsilon_h(\delta)}\le Ch.
\]
Then
\[
    (T_h^\lambda\varphi)_q
    =\varphi_q+h\lambda\log\int_A
    \exp\left(\frac{H_q(\varphi,\delta)+\varepsilon_h(\delta)}{\lambda}\right)
    \nu(d\delta).
\]
The log-sum-exp map is one-Lipschitz under uniform perturbations of the
integrand. Thus the difference from \(\varphi_q+hH_q^\lambda(\varphi)\) is at
most \(Ch^2\), uniformly in \(\lambda\).
\end{proof}

\begin{lemma}[Measurability of Gibbs kernels]\label{lem:gibbs-measurable}
For every \(\varphi\in\R^{2Q+1}\), \(q\in\cQ\), and \(h>0\), the map
\[
    \delta\mapsto C_h^\delta\varphi(q)
\]
is continuous on \(A\). Consequently, the exact Bellman kernels
\eqref{eq:gibbs-exact} and the Hamiltonian-Gibbs kernels \eqref{eq:gibbs-ham},
extended piecewise constantly in time, are Borel stochastic kernels from
\([0,T]\times\cQ\) to \(A\).
\end{lemma}

\begin{proof}
By \Cref{lem:finite-state-fk},
\[
    C_h^\delta\varphi(q)
    =
    -\frac1\gamma\log (e^{hK^\delta}u)_q,
    \qquad u_r=e^{-\gamma\varphi_r}.
\]
The entries of \(K^\delta\) are continuous in \(\delta\), and the matrix
exponential is continuous in finite dimension. The expectation representation
in \eqref{eq:finite-state-FK} gives \((e^{hK^\delta}u)_q>0\). Hence
\(\delta\mapsto C_h^\delta\varphi(q)\) is continuous. The Hamiltonian-Gibbs
kernel is Borel because \(\delta\mapsto H_q(y,\delta)\) is continuous. The
denominators in both Gibbs densities are strictly positive, and the
piecewise-constant time extension over the grid intervals is Borel.
\end{proof}

\begin{lemma}[Monotonicity and nonexpansiveness]\label{lem:nonexpansive}
The operator \(T_h^\lambda\) is monotone and translation invariant. Hence
\[
    \norm{T_h^\lambda\varphi-T_h^\lambda\psi}_\infty
    \le
    \norm{\varphi-\psi}_\infty.
\]
\end{lemma}

\begin{proof}
The certainty equivalent is increasing and translation equivariant, and the
log-sum-exp operation preserves both properties. Hence \(T_h^\lambda\) is
monotone and translation invariant. If \(\norm{\varphi-\psi}_\infty\le r\), then
\(\psi-r\mathbf 1\le\varphi\le\psi+r\mathbf 1\). Applying monotonicity and
translation invariance gives
\[
    T_h^\lambda\psi-r\mathbf 1
    \le T_h^\lambda\varphi
    \le T_h^\lambda\psi+r\mathbf 1,
\]
which is the contraction estimate.
\end{proof}

\subsection{Completion of the main proof}

\begin{proof}[Proof of \Cref{thm:main}]
Let
\[
    E_n:=\norm{v_n^{h,\lambda}-v^\lambda(t_n)}_\infty.
\]
By \Cref{lem:nonexpansive},
\[
\begin{aligned}
    E_n
    &\le
    E_{n+1}
    +\norm{T_h^\lambda v^\lambda(t_{n+1})-v^\lambda(t_n)}_\infty.
\end{aligned}
\]
By \Cref{lem:bounds,lem:consistency},
\[
    T_h^\lambda v^\lambda(t_{n+1})
    =v^\lambda(t_{n+1})+hH^\lambda(v^\lambda(t_{n+1}))+O(h^2).
\]
Since \(v^\lambda\) solves the soft ODE and \(H^\lambda\) is bounded and
Lipschitz on the bounded region,
\[
    v^\lambda(t_n)
    =v^\lambda(t_{n+1})+hH^\lambda(v^\lambda(t_{n+1}))+O(h^2).
\]
Therefore \(E_n\le E_{n+1}+Ch^2\). Since \(E_N=0\),
\[
    \max_{0\le n\le N}E_n\le Ch.
\]
Combining this with \Cref{lem:soft-hard} proves \eqref{eq:main-rate}. The
value estimate \eqref{eq:value-rate} follows from \Cref{lem:verification}.
\end{proof}

\section{Proofs for Fresh-Sampling Policy Performance}
\label{app:policy-proofs}

\begin{lemma}[Realization of fresh-sampling relaxed policies]
\label{lem:fresh-realization}
Let \((u,q)\mapsto\pi_{u,q}\) be a Borel stochastic kernel from
\([t,T]\times\cQ\) to \(A\). Then there exists a filtered probability space
carrying a Brownian motion \(W\) and marked point measures \(M^a,M^b\) such that
the compensators are given by \eqref{eq:fresh-compensators}, the above cash and
inventory dynamics hold, and \(M^a\) and \(M^b\) have no common jumps.
\end{lemma}

\begin{proof}
Since \(A\) is compact metric, there exists a Borel map
\(G:[t,T]\times\cQ\times[0,1]\to A\) such that, if \(U\) is uniform on
\([0,1]\), then \(G(u,q,U)\sim\pi_{u,q}\). Let \(G=(G^a,G^b)\). Take a
Brownian motion \(W\) and two independent Poisson random measures
\(\mathcal N^a,\mathcal N^b\) on \([t,T]\times[0,1]\times[0,\infty)\), each
with intensity \(du\,d\xi\,dz\), independent of \(W\).

For every Borel set \(B\subset A\), define recursively
\[
\begin{aligned}
    M^a([t,u]\times B)
    &:=
    \int_t^u\int_0^1\int_0^\infty
    \one_{\{G(r,q_{r-},\xi)\in B\}}
    \one_{\{z\le
        \Lambda^a(G^a(r,q_{r-},\xi))\one_{\{q_{r-}>-Q\}}\}}
    \mathcal N^a(dr,d\xi,dz),\\
    M^b([t,u]\times B)
    &:=
    \int_t^u\int_0^1\int_0^\infty
    \one_{\{G(r,q_{r-},\xi)\in B\}}
    \one_{\{z\le
        \Lambda^b(G^b(r,q_{r-},\xi))\one_{\{q_{r-}<Q\}}\}}
    \mathcal N^b(dr,d\xi,dz).
\end{aligned}
\]
The accepted intensities are bounded by \(\overline\Lambda\), so the recursive
construction is nonexplosive on \([t,T]\). Hence the finite-state jump equation
has a pathwise unique nonexplosive solution driven by the accepted proposals.
For example, the compensator of \(M^a(\cdot,B)\) is
\[
    \int_t^u\int_0^1
    \one_{\{G(r,q_{r-},\xi)\in B\}}
    \Lambda^a(G^a(r,q_{r-},\xi))
    \one_{\{q_{r-}>-Q\}}
    d\xi\,dr
    =
    \int_t^u\int_B
    \one_{\{q_{r-}>-Q\}}\Lambda^a(\delta^a)
    \pi_{r,q_{r-}}(d\delta)\,dr .
\]
The bid compensator is identical. Independence and diffuse time intensity of
\(\mathcal N^a,\mathcal N^b\) imply that ask and bid accepted jumps do not
occur at the same time almost surely.
\end{proof}

\begin{lemma}[Evaluation of a fresh-sampling relaxed policy]\label{lem:relaxed-verification}
Fix a fresh-sampling relaxed Markov policy \(\pi\). Let \(u^\pi\) solve
\begin{equation}\label{eq:relaxed-policy-ode}
    -\dot u_q^\pi(t)
    =
    \int_AH_q(u^\pi(t),\delta)\pi_{t,q}(d\delta),
    \qquad
    u_q^\pi(T)=-\Phi q^2.
\end{equation}
Then
\[
    J(t,x,s,q;\pi)=x+qs+u_q^\pi(t).
\]
Moreover, \(u^\pi\) is uniformly bounded over all such policies.
\end{lemma}

\begin{proof}
The averaged Hamiltonian
\[
    \bar H_q^\pi(y,t):=\int_AH_q(y,\delta)\pi_{t,q}(d\delta)
\]
is locally Lipschitz in \(y\) on bounded sets. The max-min argument in
\Cref{lem:bounds} applies to \(\bar H^\pi\), so the ODE has a global uniformly
bounded solution.

Set \(U(t,x,s,q):=x+qs+u_q^\pi(t)\) and
\[
    M_u:=
    \exp\left(
        -\gamma\left[
            U(u,X_u,S_u,q_u)-\int_t^u\eta q_r^2\,dr
        \right]
    \right).
\]
Let \(q=q_{u-}\). At a marked ask fill with mark \(\delta\), the jump of
\(U\) is \(\Delta_q^a(u^\pi(u),\delta)\), and hence
\[
    \frac{M_u}{M_{u-}}
    =
    \exp\left(-\gamma\Delta_q^a(u^\pi(u),\delta)\right).
\]
At a marked bid fill, the corresponding ratio is
\(\exp(-\gamma\Delta_q^b(u^\pi(u),\delta))\). Therefore the predictable jump
drift under \eqref{eq:fresh-compensators} is the \(\pi_{u,q}\)-average of the
deterministic quote jump drifts:
\[
\begin{aligned}
    &M_{u-}\int_A
    \left(e^{-\gamma\Delta_q^a(u^\pi(u),\delta)}-1\right)
    \Lambda_q^a(\delta^a)\pi_{u,q}(d\delta)\,du \\
    &\quad
    +M_{u-}\int_A
    \left(e^{-\gamma\Delta_q^b(u^\pi(u),\delta)}-1\right)
    \Lambda_q^b(\delta^b)\pi_{u,q}(d\delta)\,du .
\end{aligned}
\]
Combining this jump drift with the running penalty and the Brownian correction
gives the averaged Hamiltonian term. The exponential It\^o calculation yields,
for a local martingale \(\mathcal M\),
\[
    dM_u
    =
    -\gamma M_{u-}
    \left[
        \dot u_{q_{u-}}^\pi(u)+
        \int_AH_{q_{u-}}(u^\pi(u),\delta)\pi_{u,q_{u-}}(d\delta)
    \right]du
    +d\mathcal M_u.
\]
The drift vanishes by \eqref{eq:relaxed-policy-ode}. As in
\Cref{lem:verification}, for every stopping time \(\tau\le T\),
\[
    M_{\tau}
    \le
    C\exp\left(
        C\left|\int_t^{\tau}q_{u-}\,dS_u\right|
        +C\bigl[M^a([t,\tau]\times A)+M^b([t,\tau]\times A)\bigr]
    \right).
\]
The total marked count has bounded intensity, so \Cref{lem:exp-moments} gives
\[
    \sup_{\tau\le T}\E[M_\tau^p]<\infty
\]
for some \(p>1\). Hence the local martingale is of class \(D\), and therefore
is a true martingale. Evaluating at \(t\) and \(T\) gives the stated identity.
\end{proof}

\subsection{Exact Bellman policy estimates}

\begin{lemma}[Local regret of the exact Bellman Gibbs policy]\label{lem:local-regret-exact}
There exist constants \(C>0\), \(h_0>0\), and \(\lambda_0>0\) such that, for
every \(h=T/N\le h_0\), \(0<\lambda\le\lambda_0\), \(0\le n\le N-1\),
\(t\in[t_n,t_{n+1})\), and \(q\in\cQ\),
\[
    H_q^0(v^0(t))
    -\int_AH_q(v^0(t),\delta)\pi_{n,q}^{\rm ex}(d\delta)
    \le
    C\left[h+\lambda(1+\abs{\log\lambda})\right].
\]
\end{lemma}

\begin{proof}
Let \(\varphi:=v_{n+1}^{h,\lambda}\). By \Cref{lem:consistency},
\[
    C_h^\delta\varphi(q)
    =\varphi_q+h\{H_q(\varphi,\delta)+\varepsilon_h(\delta)\},
    \qquad
    \sup_{\delta\in A}\abs{\varepsilon_h(\delta)}\le Ch.
\]
The constant \(\varphi_q/(h\lambda)\) cancels from the Gibbs density, so
\(\pi_{n,q}^{\rm ex}\) is the Gibbs measure for
\(\widetilde H_q(\delta):=H_q(\varphi,\delta)+\varepsilon_h(\delta)\). The Gibbs
variational formula implies
\[
    \int_A\widetilde H_q(\delta)\pi_{n,q}^{\rm ex}(d\delta)
    \ge
    \lambda\log\int_Ae^{\widetilde H_q(\delta)/\lambda}\nu(d\delta).
\]
Since log-sum-exp is one-Lipschitz under uniform perturbations,
\[
    \int_AH_q(\varphi,\delta)\pi_{n,q}^{\rm ex}(d\delta)
    \ge H_q^\lambda(\varphi)-Ch.
\]
Using \Cref{lem:entropy-bias},
\[
    H_q^0(\varphi)-\int_AH_q(\varphi,\delta)\pi_{n,q}^{\rm ex}(d\delta)
    \le C\lambda(1+\abs{\log\lambda})+Ch.
\]
Finally, \Cref{thm:main,lem:time-lip} gives
\[
    \norm{\varphi-v^0(t)}_\infty
    \le C\left[h+\lambda(1+\abs{\log\lambda})\right].
\]
Transfer the estimate by \Cref{lem:lip}.
\end{proof}

\begin{proof}[Proof of \Cref{cor:performance-exact}]
Let \(u^{\rm ex}\) solve the averaged policy ODE for \(\pi^{\rm ex}\). By
\Cref{lem:relaxed-verification},
\[
    J(t,x,s,q;\pi^{\rm ex})=x+qs+u_q^{\rm ex}(t).
\]
Nonnegativity follows by applying the verification
submartingale argument to \(v^0\) under the relaxed policy, since the averaged
Hamiltonian is bounded above by \(H^0\).

For the upper bound, let \(E(t):=\norm{v^0(t)-u^{\rm ex}(t)}_\infty\). The
integral forms and \Cref{lem:local-regret-exact,lem:lip} give
\[
    E(t)
    \le
    \int_t^T
    \left[
        C\{h+\lambda(1+\abs{\log\lambda})\}+L E(r)
    \right]dr.
\]
Backward Gronwall proves the claim.
\end{proof}

\subsection{Hamiltonian-Gibbs proxy estimates}

\begin{lemma}[Soft-HJB Euler error]\label{lem:soft-euler-error}
There exist constants \(C>0\) and \(h_0>0\) such that, for every
\(h=T/N\le h_0\) and \(0<\lambda\le1\),
\[
    \max_{0\le n\le N}\norm{\widehat v_n^{h,\lambda}-v^0(t_n)}_\infty
    \le
    C\left[h+\lambda(1+\abs{\log\lambda})\right].
\]
\end{lemma}

\begin{proof}
By \Cref{lem:bounds,lem:lip}, \(H^\lambda\) is uniformly bounded and uniformly
Lipschitz on the bounded region reached by \(v^\lambda\), with constants
independent of \(0<\lambda\le1\). Since \(H_q(y,\delta)\) is \(C^1\) in \(y\)
with uniformly bounded derivatives on bounded sets, the standard global Euler
estimate for Lipschitz ODEs is uniform in \(\lambda\). For \(h\le h_0\), the
Euler iterates remain in a slightly larger bounded region and
\[
    \max_{0\le n\le N}
    \norm{\widehat v_n^{h,\lambda}-v^\lambda(t_n)}_\infty
    \le Ch .
\]
Combining this with \Cref{lem:soft-hard} gives the claimed estimate against
\(v^0\).
\end{proof}

\begin{lemma}[Local regret of the Hamiltonian-Gibbs proxy]\label{lem:local-regret-ham}
There exist constants \(C>0\), \(h_0>0\), and \(\lambda_0>0\) such that, for
every \(h=T/N\le h_0\), \(0<\lambda\le\lambda_0\), \(0\le n\le N-1\),
\(t\in[t_n,t_{n+1})\), and \(q\in\cQ\),
\[
    H_q^0(v^0(t))
    -\int_AH_q(v^0(t),\delta)\pi_{n,q}^{\rm Ham}(d\delta)
    \le
    C\left[h+\lambda(1+\abs{\log\lambda})\right].
\]
\end{lemma}

\begin{proof}
Let \(\varphi:=\widehat v_{n+1}^{h,\lambda}\). By the Gibbs variational formula
applied directly to \(H_q(\varphi,\cdot)\),
\[
    \int_AH_q(\varphi,\delta)\pi_{n,q}^{\rm Ham}(d\delta)
    \ge
    H_q^\lambda(\varphi).
\]
Therefore
\[
    H_q^0(\varphi)-\int_AH_q(\varphi,\delta)\pi_{n,q}^{\rm Ham}(d\delta)
    \le C\lambda(1+\abs{\log\lambda})
\]
by \Cref{lem:entropy-bias}. By \Cref{lem:soft-euler-error,lem:time-lip},
\[
    \norm{\varphi-v^0(t)}_\infty
    \le C\left[h+\lambda(1+\abs{\log\lambda})\right].
\]
The Lipschitz estimate transfers the regret bound from \(\varphi\) to
\(v^0(t)\).
\end{proof}

\begin{proof}[Proof of \Cref{prop:ham-performance}]
Let \(u^{\rm Ham}\) solve the averaged policy ODE for \(\pi^{\rm Ham}\). By
\Cref{lem:relaxed-verification},
\[
    J(t,x,s,q;\pi^{\rm Ham})=x+qs+u_q^{\rm Ham}(t).
\]
As in \Cref{cor:performance-exact}, the verification submartingale for \(v^0\)
under the relaxed policy gives nonnegativity of the performance gap.

For the upper bound, let \(E(t):=\norm{v^0(t)-u^{\rm Ham}(t)}_\infty\). Using
the integral forms of the hard ODE and the policy-evaluation ODE,
\Cref{lem:local-regret-ham}, and the Lipschitz estimate in \Cref{lem:lip}, we
obtain
\[
    E(t)
    \le
    \int_t^T
    \left[
        C\{h+\lambda(1+\abs{\log\lambda})\}+L E(r)
    \right]dr.
\]
Backward Gronwall proves the claimed bound.
\end{proof}

\section{Proofs for Quote Concentration}
\label{app:quote-proofs}

\begin{proof}[Proof of \Cref{lem:exp-curvature}]
For the ask side, write
\[
    d_q^a(t):=v_{q-1}^0(t)-v_q^0(t),
\]
and define the one-dimensional contribution
\[
    f_a(x;d):=\frac{\alpha_a}{\gamma}e^{-k_ax}
    \left(1-e^{-\gamma(x+d)}\right).
\]
Then
\[
    \frac{\partial^2 f_a}{\partial x^2}(x;d)
    =
    \frac{\alpha_a}{\gamma}e^{-k_ax}
    \left[k_a^2-(k_a+\gamma)^2e^{-\gamma(x+d)}\right].
\]
For every \(x\in[\underline\delta,\overline\delta]\) and every active ask state,
\[
    x+d_q^a(t)
    \le
    \overline\delta+
    \sup_{t,q>-Q}(v_{q-1}^0(t)-v_q^0(t))
    =D_a.
\]
Since \(D_a<\Theta_a\),
\[
    (k_a+\gamma)^2e^{-\gamma D_a}-k_a^2>0.
\]
Thus
\[
    \frac{\partial^2 f_a}{\partial x^2}(x;d_q^a(t))\le-\mu_a
\]
uniformly on the quote interval. The bid side is identical, with
\(d_q^b(t):=v_{q+1}^0(t)-v_q^0(t)\), and yields
\[\partial_{xx}f_b\le-\mu_b.\]

For a \(\mu_i\)-strongly concave function on a compact interval, any constrained
maximizer \(x^*\) satisfies
\[
    f_i(x^*)-f_i(x)\ge\frac{\mu_i}{2}\abs{x-x^*}^2.
\]
Indeed, strong concavity gives
\(f_i(x)\le f_i(x^*)+f_i'(x^*)(x-x^*)-\mu_i\abs{x-x^*}^2/2\), and the
first-order condition for a constrained maximum gives
\(f_i'(x^*)(x-x^*)\le0\).

The Hamiltonian is the sum of a constant, the active ask contribution, and the
active bid contribution. Inactive coordinates do not enter \(H_q\) and are
removed by \(P_q\). Summing the active one-dimensional inequalities gives
Assumption~\ref{ass:qg} with \(\mu=\min\{\mu_a,\mu_b\}\).
\end{proof}

\begin{proof}[Proof of \Cref{cor:quote-concentration}]
For \(\pi^{\rm ex}\), use \Cref{lem:local-regret-exact}; for
\(\pi^{\rm Ham}\), use \Cref{lem:local-regret-ham}. In either case,
\[
    H_q^0(v^0(t))-
    \int_AH_q(v^0(t),\delta)\pi_{n,q}(d\delta)
    \le
    C\left[h+\lambda(1+\abs{\log\lambda})\right].
\]
By Assumption~\ref{ass:qg},
\[
    \frac\mu2
    \int_A\dist_q^2(\delta,\cA_q^*(t))\pi_{n,q}(d\delta)
    \le
    H_q^0(v^0(t))-
    \int_AH_q(v^0(t),\delta)\pi_{n,q}(d\delta).
\]
This proves concentration. If the active optimizer is unique, Jensen's
inequality gives
\[
    \abs{P_q(\bar\delta_{n,q}^{\pi}-\delta_q^*(t))}^2
    \le
    \int_A\abs{P_q(\delta-\delta_q^*(t))}^2\pi_{n,q}(d\delta),
\]
and the remaining statements follow by integrating and summing.
\end{proof}

\section{Numerical Implementation and Reproducibility Details}
\label{app:numerical-details}

\subsection{Additional numerical tables and checks}

\begin{table}[t]
\centering
\caption{Financial diagnostics under fresh-sampling marked-Poisson simulation
with 5000 paths and common random numbers.}
\label{tab:financial-diagnostics}
\resizebox{\textwidth}{!}{%
\begin{tabular}{lrrrrrr}
\toprule
Strategy & CE & Mean reward & Std reward & Sharpe-like
& Mean terminal PnL & Time-avg. \(q^2\) \\
\midrule
Hard optimal & 0.681759 & 0.705764 & 0.701899 & 1.005507
& 0.729234 & 0.525851 \\
Gibbs proxy & 0.680443 & 0.703603 & 0.689147 & 1.020977
& 0.728519 & 0.560121 \\
Constant spread & 0.564824 & 0.576441 & 0.485755 & 1.186690
& 0.616601 & 0.888772 \\
Inventory-linear & 0.566948 & 0.578514 & 0.484631 & 1.193720
& 0.615227 & 0.837837 \\
\bottomrule
\end{tabular}%
}
\end{table}

\begin{table}[t]
\centering
\caption{Active-coordinate quote convergence and Hamiltonian regret.}
\label{tab:quote-results}
\begin{tabular}{ccccc}
\toprule
\(h\) & \(\lambda\) & theoretical scale
& squared quote error & Hamiltonian regret \\
\midrule
0.02000 & 0.050 & 0.219787 & 0.238605 & 0.418369 \\
0.01000 & 0.020 & 0.108240 & 0.088551 & 0.188497 \\
0.00500 & 0.010 & 0.061052 & 0.039221 & 0.102124 \\
0.00250 & 0.005 & 0.033992 & 0.016583 & 0.054804 \\
0.00125 & 0.002 & 0.015679 & 0.004985 & 0.023593 \\
\bottomrule
\end{tabular}
\end{table}

\begin{table}[t]
\centering
\caption{Exact Bellman one-step consistency.}
\label{tab:exact-consistency}
\begin{tabular}{cccc}
\toprule
\(h\) & consistency error & error/\(h^2\) & fitted slope \\
\midrule
0.05000 & \(2.22\cdot10^{-4}\) & 0.088670 & 1.9873 \\
0.02500 & \(5.63\cdot10^{-5}\) & 0.090029 & 1.9873 \\
0.01250 & \(1.42\cdot10^{-5}\) & 0.090721 & 1.9873 \\
0.00625 & \(3.56\cdot10^{-6}\) & 0.091071 & 1.9873 \\
\bottomrule
\end{tabular}
\end{table}

\begin{table}[t]
\centering
\caption{Exact Bellman Gibbs policy versus Hamiltonian-Gibbs proxy.}
\label{tab:exact-vs-proxy}
\begin{tabular}{cccccc}
\toprule
\(h\) & value gap & quote gap sq. & exact CE gap
& proxy CE gap & exact-proxy CE gap \\
\midrule
0.05000 & \(7.13\cdot10^{-4}\) & \(1.71\cdot10^{-6}\)
& 0.003801 & 0.003812 & \(1.07\cdot10^{-5}\) \\
0.02500 & \(3.50\cdot10^{-4}\) & \(3.71\cdot10^{-7}\)
& 0.003804 & 0.003809 & \(5.46\cdot10^{-6}\) \\
0.01250 & \(1.73\cdot10^{-4}\) & \(8.62\cdot10^{-8}\)
& 0.003805 & 0.003808 & \(2.76\cdot10^{-6}\) \\
0.00625 & \(8.60\cdot10^{-5}\) & \(2.08\cdot10^{-8}\)
& 0.003806 & 0.003807 & \(1.39\cdot10^{-6}\) \\
\bottomrule
\end{tabular}
\end{table}

\subsection{Robustness and numerical accuracy checks}
\label{subsec:robustness-checks}

We performed two additional checks.

First, we checked the Gauss--Legendre quadrature error in the action integral.
Here \(n_{\rm quad}\) denotes the number of one-dimensional Gauss--Legendre
nodes. Using 321 quadrature nodes as the reference, the maximum error over
\[
    n_{\rm quad}\in\{21,41,81,161\}
\]
is
\[
    8.20\cdot 10^{-13}.
\]
Thus action quadrature error is negligible relative to the main convergence
errors.

Second, we tested sensitivity to the risk-aversion parameter. Fixing
\[
    h=0.0025,
    \qquad
    \lambda=0.005,
\]
we swept
\[
    \gamma\in\{0.01,0.02,0.05,0.10,0.20,0.50,1,2,5,10\}.
\]
For each \(\gamma\), we computed
\[
    C_{\mathrm{emp}}^{\mathrm{value}}(\gamma)
    :=
    \frac{\|\widehat v^{h,\lambda}-v^0\|_\infty}
    {h+\lambda(1+|\log\lambda|)}
\]
and
\[
    C_{\mathrm{emp}}^{\mathrm{policy}}(\gamma)
    :=
    \frac{V^*-CE(\pi^{\rm Ham})}
    {h+\lambda(1+|\log\lambda|)}.
\]
The empirical constants remain stable for moderate risk aversion and rise in
parts of the high-risk-aversion regime. The maximum observed value constant is
\[
    \max_\gamma C_{\mathrm{emp}}^{\mathrm{value}}(\gamma)=0.8726.
\]
We use this only as a robustness check and do not claim monotonicity of the
hidden constant in \(\gamma\).

\subsection{Implementation details}

All reported value errors are computed on the grid
\[
    \{t_n=nh:0\le n\le N\}\times\cQ
\]
using the sup norm
\[
    \|e\|_\infty:=\max_{0\le n\le N}\max_{q\in\cQ}|e_{n,q}|.
\]
If \(e_q(t_n)\) denotes the corresponding quote or regret error at grid point
\((t_n,q)\), time-integrated quote and regret errors are computed by the left
Riemann rule
\[
    \sum_{q\in\cQ}\sum_{n=0}^{N-1}h\,e_q(t_n).
\]

For the continuous hard and soft ODEs, we solve the forward equations in
\(\tau=T-t\). Thus \(u(\tau)=v(T-\tau)\) satisfies
\[
    \dot u(\tau)=H(u(\tau)),\qquad u(0)=(-\Phi q^2)_{q\in\cQ},
\]
with \(H=H^0\) or \(H=H^\lambda\). The reference solution is computed by the
classical fourth-order Runge--Kutta method with step size
\(h_{\rm ref}=10^{-3}\).

For the uniform reference law on
\(A=[\underline\delta,\overline\delta]^2\), let
\(\{(\xi_i,\omega_i)\}_{i=1}^{n_{\rm quad}}\) be Gauss--Legendre nodes and
weights on \([-1,1]\), and set
\[
    d_i=\frac{\overline\delta+\underline\delta}{2}
        +\frac{\overline\delta-\underline\delta}{2}\xi_i .
\]
Then
\[
    \int_A f(\delta)\nu(d\delta)
    \approx
    \frac14\sum_{i=1}^{n_{\rm quad}}\sum_{j=1}^{n_{\rm quad}}
    \omega_i\omega_j f(d_i,d_j).
\]
For each \(q\in\cQ\), set
\[
    m_q:=\max_{i,j}H_q(y,(d_i,d_j)).
\]
The soft Hamiltonian is evaluated by the stabilized log-sum-exp formula
\[
    H_q^\lambda(y)
    =
    m_q+\lambda\log\left[
    \frac14\sum_{i,j}\omega_i\omega_j
    \exp\left(
        \frac{H_q(y,(d_i,d_j))-m_q}{\lambda}
    \right)\right].
\]
Unless otherwise stated, \(n_{\rm quad}=61\). The quadrature check uses
\(n_{\rm quad}=321\) as the reference. Exact Bellman validation uses
\(n_{\rm exact}=17\) Gauss--Legendre nodes per coordinate for the action
integral in \(T_h^\lambda\).

For exact Bellman validation, \(C_h^\delta\varphi(q)\) is computed from
\[
    C_h^\delta\varphi(q)
    =
    -\frac1\gamma
    \log\left[
        \left(e^{hK^\delta}u\right)_q
    \right],
    \qquad
    u_q=e^{-\gamma\varphi_q},
\]
with \(K^\delta\) defined in \eqref{eq:Kdelta}. The dense matrix exponential is
computed directly on the \((2Q+1)\)-dimensional inventory state space.

Monte Carlo diagnostics use \(N_{\rm MC}=5000\) paths and common random numbers
across strategies. All Monte Carlo tables are generated with the pseudo-random
seed \(s_{\rm MC}=12345\). For each path, all strategies use the same
Brownian shocks and the same dominating Poisson proposal streams. Ask and bid
proposals are generated independently with dominating rate
\(\overline\Lambda\). Conditional on a proposal at time \(u\) and inventory
\(q_{u-}\), a quote mark is sampled from the strategy kernel
\(\pi_{u,q_{u-}}\), and the proposal is accepted with probability
\[
    \frac{\Lambda^a(\delta^a)}{\overline\Lambda}\one_{\{q_{u-}>-Q\}}
    \quad\text{or}\quad
    \frac{\Lambda^b(\delta^b)}{\overline\Lambda}\one_{\{q_{u-}<Q\}},
\]
respectively. This implements the fresh-sampling compensators in
\eqref{eq:fresh-compensators}.

\end{document}